\documentclass[12pt,english,american]{article}
\usepackage[T1]{fontenc}
\usepackage[latin9]{inputenc}
\usepackage{geometry}
\usepackage{color}
\usepackage{babel}
\usepackage[normalem]{ulem}
\usepackage{verbatim}
\usepackage{refstyle}
\usepackage{amsthm,bbm}
\usepackage{amsmath}
\usepackage{amssymb}
\usepackage[unicode=true,pdfusetitle,
 bookmarks=true,bookmarksnumbered=false,bookmarksopen=false,
 breaklinks=true,pdfborder={0 0 0},backref=false,colorlinks=true]
 {hyperref}
\hypersetup{
 linkcolor=blue,citecolor=blue, urlcolor=blue}

\makeatletter

\AtBeginDocument{\providecommand\defref[1]{\ref{def:#1}}}
\AtBeginDocument{}
\AtBeginDocument{\providecommand\thmref[1]{\ref{thm:#1}}}
\AtBeginDocument{\providecommand\lemref[1]{\ref{lem:#1}}}
\newcommand{\noun}[1]{\textsc{#1}}
\RS@ifundefined{subref}
  {\def\RSsubtxt{section~}\newref{sub}{name = \RSsubtxt}}
  {}
\RS@ifundefined{thmref}
  {\def\RSthmtxt{theorem~}\newref{thm}{name = \RSthmtxt}}
  {}
\RS@ifundefined{lemref}
  {\def\RSlemtxt{lemma~}\newref{lem}{name = \RSlemtxt}}
  {}

\usepackage{enumitem}		
\usepackage[color=green!40]{todonotes}
\theoremstyle{plain}
\newtheorem{thm}{\protect\theoremname}[section]
\theoremstyle{remark}
\newtheorem{rem}[thm]{\protect\remarkname}
\theoremstyle{definition}
\newtheorem{defn}[thm]{\protect\definitionname}
\theoremstyle{plain}
\newtheorem{lem}[thm]{\protect\lemmaname}
\ifx\proof\undefined
\newenvironment{proof}[1][\protect\proofname]{\par
\normalfont\topsep6\p@\@plus6\p@\relax
\trivlist
\itemindent\parindent
\item[\hskip\labelsep
\scshape
#1]\ignorespaces
}{%
\endtrivlist\@endpefalse
}
\providecommand{\proofname}{Proof}
\fi

\usepackage{txfonts,refstyle}
\newref{con}{name = Conjecture\ }
\newref{prop}{name = Proposition\ }
\newref{def}{name = Definition\ }
\newref{sec}{name = Section\ }
\newref{thm}{name = Theorem\ }
\newref{lem}{name = Lemma\ }
\newref{cor}{name = Corollary\ }

\makeatother

\addto\captionsamerican{\renewcommand{\definitionname}{Definition}}
\addto\captionsamerican{\renewcommand{\lemmaname}{Lemma}}
\addto\captionsamerican{\renewcommand{\remarkname}{Remark}}
\addto\captionsamerican{\renewcommand{\theoremname}{Theorem}}
\addto\captionsenglish{\renewcommand{\definitionname}{Definition}}
\addto\captionsenglish{\renewcommand{\lemmaname}{Lemma}}
\addto\captionsenglish{\renewcommand{\remarkname}{Remark}}
\addto\captionsenglish{\renewcommand{\theoremname}{Theorem}}
\providecommand{\definitionname}{Definition}
\providecommand{\lemmaname}{Lemma}
\providecommand{\remarkname}{Remark}
\providecommand{\theoremname}{Theorem}

\begin{document}

\title{Effective Dynamics\\
of a Tracer Particle Interacting\\
with an Ideal Bose Gas}

\author{D.-A. Deckert, J. Fröhlich, P. Pickl, A. Pizzo}

\maketitle
\begin{abstract}
We study a system consisting of a heavy quantum particle, called
\emph{tracer particle,} coupled to an ideal gas of light Bose
particles, the ratio of masses of the tracer particle and
 a gas particle being proportional to the gas density.
All particles have non-relativistic kinematics.
The tracer particle is driven by an external potential and couples
to the gas particles through a pair potential. We compare the
quantum dynamics of this system to an effective dynamics given by a Newtonian equation of motion for the tracer particle coupled to a classical wave equation for the Bose gas. We quantify the closeness of these two dynamics as the mean-field limit is approached (gas density $\to\infty$). Our estimates allow us to interchange the thermodynamic with the mean-field limit.
\end{abstract}

\section{Introduction}

As a model for the dynamics of a heavy particle that interacts with
an ideal Bose gas of light particles we consider the Schrödinger equation

\selectlanguage{english}%
\begin{equation}
i\partial_{t}\Psi_{t}(x,y_{1,\ldots,}y_{N})=H\Psi_{t}(x,y_{1,\ldots,}y_{N}),\label{eq:microscopic}
\end{equation}
with
\begin{equation}
H:=-\frac{\Delta_{x}}{2\rho}+\rho V(x)-\sum_{k=1}^{N}\Delta_{y_{k}}+\sum_{k=1}^{N}W(x-y_{k}),\label{eq:micro hamiltonian}
\end{equation}
\foreignlanguage{american}{for vectors $\Psi_{t}$ in the Hilbert
space of complex-valued, square-integrable $(N+1)$-particle wave
functions on configuration space $\mathbb{R}^{3(N+1)}$, i.e.,
\[
{\cal H}:=L^{2}(\mathbb{R}^{3(N+1)},\mathbb{C}).
\]
Here $x\in\mathbb{R}^{3}$ represents the position of the tracer particle, and $y_{1},\ldots,y_{N}\in\mathbb{R}^{3}$
are the positions of the $N$ gas particles. The Laplace operators w.r.t. $x$
and $y_{k}$ are denoted by $\Delta_{x}$ and $\Delta_{y_{k}},$ respectively.
The potential $W$ describes the interaction of the tracer particle
with the gas, and the potential $V$ describes an external
force driving the tracer particle. The support of
each gas particle tensor component of $\Psi_{t}|_{t=0}$ is given by a compact region $\Lambda\subset\mathbb{R}^{3}$
the volume of which is denoted by $|\Lambda|$. The mean particle density is given by
\begin{equation}
\rho:=\frac{N}{|\Lambda|}\label{eq:rho_def}
\end{equation}
and is kept fixed. The ratio between the mass of the tracer particle
and the one of the gas particle is chosen to be $2\rho$. We use units such that $\hbar=1$ -- note that $V$ has different dimension than $W$.}\\

\selectlanguage{american}%
One might expect that if the tracer particle is well-localized initially, and because of
its heavy mass, the variance of its position remains small. Therefore, its mean position
should effectively obey Newton's second law, for a force generated
by the external potential $V$ and by the mean-field version of the
potential $W$ originating from an effective dynamics of the gas.
It should be observed that, for large volume $\left|\Lambda\right|$,
the influence of the tracer particle on the gas is not of leading
order. Yet the back reaction of the perturbation of the
gas on the tracer particle is relevant. In our study, we are interested
in the (one-particle) wave function of the gas, as compared to a reference state $\phi_{t}^{(\mathrm{ref})}$ that obeys
\begin{equation}
i\partial_{t}\phi_{t}^{(\mathrm{ref})}(y)=-\Delta_{y}\phi_{t}^{(\mathrm{ref})}(y)\label{eq:free_reference}
\end{equation}
and describes the effective one-particle wave function of a system
of gas particles in the absence of the tracer particle. The effective
one-particle wave function of the system of  gas particles interacting
with the tracer particle is denoted by
\[
\phi_{t}^{(\mathrm{ref})}+\epsilon_{t}.
\]
The function $\epsilon_t$ describes an "excited state" of the gas.

For large $|\Lambda|$ and $\rho$, the Schr\"odinger dynamics given by 
(\ref{eq:microscopic}), (\ref{eq:micro hamiltonian}) is expected to be well approximated by a semi-classical dynamics described by the equations

\begin{eqnarray}
i\partial_{t}\epsilon_{t}(y) & = & \left(-\Delta_{y}+W\left(X_{t}-y\right)\right)\epsilon_{t}(y)+W(X_{t}-y)\phi_{t}^{(\mathrm{ref})}(y),\label{eq:macro_gas}\\
\frac{d^{2}X_{t}}{dt^{2}} & = & -\nabla V(X_{t})-\nabla W*|\epsilon_{t}|^{2}(X_{t})-2\Re\nabla W*\left(\overline{\phi_{t}^{(\mathrm{ref})}}\epsilon_{t}\right)(X_{t}).\label{eq:macro_tracer}
\end{eqnarray}
Here $t \mapsto X_t \in \mathbb{R}^3$ denotes a classical trajectory of the tracer particle. Equation (\ref{eq:macro_gas}) is a Hartree-type equation, and $\epsilon_t$ is called Hartree (one-particle) wave function. In order to keep the excited state of the gas described by $\epsilon_t$ visible in the scaling considered in this paper, we choose the norm of $\phi_{t}^{(\mathrm{ref})}$ such that the inhomogeneity in (\ref{eq:macro_gas}) is of order $O(1)$, i.e.,
\begin{equation}
\left\Vert \phi_{t=0}^{(\mathrm{ref})}\right\Vert _{\infty}={\cal O}(1).\label{eq:ref_nom}
\end{equation}

\begin{rem}
In our case $\phi_{t}^{(\mathrm{ref})}$ varies only little in a neighborhood
of $X_{t}$. Hence, it is possible to replace the inhomogeneous term
in (\ref{eq:macro_gas}) by
\[
W(X_{t}-y)\phi_{t}^{(\mathrm{ref})}(X_{t}).
\]

\end{rem}
From now on, we refer to the time evolution generated
by (\ref{eq:microscopic}) as the \emph{microscopic dynamics} and to the
one generated by (\ref{eq:macro_gas})-(\ref{eq:macro_tracer}) as the
\emph{macroscopic dynamics}. The goal of this work is to quantify
the closeness of these two dynamics and to estimate the rate of convergence, as $\left|\Lambda\right|,\rho\to\infty$.\\

Derivations of such mean-field equations from the microscopic $N$-body
Schrödinger evolution are usually carried out by making use of reduced density matrices and are based on hierarchies \cite{spohn,erdos}.
In recent years, alternative methods have been developed
to derive the Hartree equation from the microscopic dynamics. One
approach, developed in \cite{froehlich}, relies on the Heisenberg picture and involves dispersive estimates and the counting of Feynman graphs. Another one was
introduced in \cite{rodnianski}. It exploits the dynamics of coherent states and is inspired by a semiclassical argument given by Hepp \cite{hepp}, who initiated all these studies. In this paper we follow a different
approach introduced in \cite{Pickl_2011_A-Simple-Derivation-of-Mean-Field-Limits-for-Quantum-Systems},
which is based on counting the number of ``bad'' particles, i.e.,
particles that are not in the state given by the Hartree wave function.

The Hartree wave function is of interest because it can be used to
investigate physically interesting phenomena, such as quantum friction
by emission of \v{C}erenkov radiation \cite{Froehlich_2011_Friction-in-a-Model-of-Hamiltonian-Dynamics}, more easily than by using the microscopic dynamics. Rigorous control of the microscopic time-evolution in terms of a macroscopic one is however a difficult problem. Our paper provides a first result in this direction for a system consisting of a tracer particle interacting with an ideal Bose gas. For the analysis of \v{C}erenkov radiation, i.e., the deceleration of a particle with a speed higher than the speed of sound in the gas, interacting Bose gases, where the speed of sound is non-zero, are most interesting. The techniques to treat the ideal gas presented in this paper appear to be
very robust and to allow for many generalizations. A mean
field pair interaction of the gas particles, for example, can readily be introduced
into our mothods, using estimates provided in \cite{Pickl_2011_A-Simple-Derivation-of-Mean-Field-Limits-for-Quantum-Systems}. Further generalizations to models with a thermodynamic scaling of interacting gases are presently under investigation.\\

\noindent\textbf{Acknowledgments:} D.-A.D. gratefully acknowledges
financial support from the post-doc program of the DAAD. J.F.'s visit at IAS is supported by 'The Fund For Math' and 'The Monell Foundation'. A.P. is supported
by the NSF grant \#DMS-0905988. Furthermore, D.-A.D. and P.P. would
like to thank the Mathematisches Institut der LMU München and the Department of Mathematics of UC Davis for their hospitality.

\section{\label{sec:notation}Notation}
\begin{enumerate}
\item The expectation value of an operator $O$ w.r.t. the microscopic
wave function $\Psi_{t}$ is denoted by 
\[
\left\langle O\right\rangle _{t}:=\left\langle \Psi_{t},O\Psi_{t}\right\rangle .
\]

\item $\left|\cdot\right|$ is the canonical norm on $\mathbb{C}^{d}$, for
any dimension $d$; $\left\Vert \cdot\right\Vert _{p}$ is the norm
on the Lebesgue space $L^{p}$, $0\leq p\leq\infty$. For operators,
$O$, on $L^{2}$ or ${\cal H}$, we denote by $\left\Vert O\right\Vert$
their operator norm. We also introduce the norm $\Vert M \Vert_{p}$ for matrix-valued functions, M, with matrix elements $M_{ij}$. It is defined by
\[
\left\Vert M\right\Vert _{p}:=\sum_{i,j}\left\Vert M_{ij}\right\Vert _{p}.
\]
 
\item The momentum operator of the tracer particle is denoted by \foreignlanguage{english}{
\[
p:=-i\nabla_{x}.
\]
}Furthermore, for $n\in\mathbb{N}$, $D\left(\Delta^{n}\right)$ denotes
the natural domain of the self-adjoint operator $\Delta^{n}$ in $L^{2}(\mathbb{R}^{3},\mathbb{C})$,
and $D\left(\Delta^{\infty}\right):=\cap_{n\in\mathbb{N}}D\left(\Delta^{n}\right)$.
Similarly, we denote by $D(H_{0})$ and $D(H_{0}^{n})$, $n\in\mathbb{N}$,
the domains in $\mathcal{H}$ of the self-adjoint operators
\[
H_{0}:=-\frac{\Delta_{x}}{2\rho}-\sum_{k=1}^{N}\Delta_{y_{k}}\qquad \text{and}\qquad H_{0}^{n},
\]
respectively, and $D(H_{0}^{\infty}):=\cap_{n\in\mathbb{N}}D(H_{0}^{n})$.
\item The Fourier transform of a function $\eta \in L^{2}$ is denoted by $\widehat{\eta}$.
\item \label{enu:projectors}Given a vector $\eta\in L^{2}$ with $\left\Vert \eta\right\Vert _{2}=1$
we denote the orthogonal projection onto $\eta$ by
\[
\left|\eta\right\rangle \left\langle \eta\right|.
\]
Furthermore, we use the notation\foreignlanguage{english}{
\begin{equation}
q_{k}^{\eta}:=1-p_{k}^{\eta},\qquad\left(p_{k}^{\eta}\Psi\right)(x,y_{1},\ldots,y_{N}):=\eta(y_{k})\int d^{3}y_{k}\,\eta^{*}(y_{k})\Psi(x,y_{1},\ldots,y_{N}),\qquad1\leq k\leq N,\label{eq:projectors}
\end{equation}
}and\foreignlanguage{english}{
\begin{eqnarray}\label{eq:qtref}
q_{t}^{(\mathrm{ref})}:=1-p_{t}^{(\mathrm{ref})},\qquad p_{t}^{(\mathrm{ref)}}:=\left|\left|\Lambda\right|^{-1/2}\phi_{t}^{(\mathrm{ref})}\right\rangle \left\langle \left|\Lambda\right|^{-1/2}\phi_{t}^{(\mathrm{ref})}\right|.
\end{eqnarray}
}

We note that, in $q_{k}^{\eta}$, the subscript $k$ always
stands for the $k$-th tensor component, while, in $q_{t}^{(\mathrm{ref)}}$,
the subscript $t$ always refers to time $t$.

\item The convolution of two functions $f,g$ on $\mathbb{R}^{3}$ is denoted
by $(f*g)(\cdot):=\int dy\, f(\cdot-y)g(y)$.
\item The statement "$F \in \mathtt{Bounds}$" refers to a continuous, monotone increasing function $F: \mathbb{R}^{+} \rightarrow \mathbb{R}^{+}$.
\item If not specified otherwise the symbol $C$ denotes a universal constant the value of which may change from one line to the next.
\end{enumerate}

\section{Main Result}

Throughout this paper we assume that the density, $\rho$, of the gas fulfills
\[
\rho>1.
\]

In order to compare the solutions of the microscopic and macroscopic dynamics, we
choose initial conditions that are compatible with each other:
\begin{defn}
\label{def:initial_values}~\\

\begin{enumerate}[label=(\roman*)]
\item As initial conditions for the microscopic dynamics (\ref{eq:microscopic})
we consider $\Psi_{t}|_{t=0}=\Psi^{(0)}$, where 
\begin{equation}
\Psi^{(0)}(x,y_{1},\ldots,y_{N}):=\chi^{(0)}(x)\prod_{k=1}^{N}\phi^{(0)}(y_{k}),\qquad\left\Vert \Psi^{(0)}\right\Vert _{2}=1,\label{eq:initial_psi}
\end{equation}
is given in terms of some unit vectors $\chi^{(0)},\phi^{(0)}\in L^{2}(\mathbb{R}^{3},\mathbb{C})$
with the properties 
\[
\chi^{(0)}\in D\left(\Delta^{\infty}\right),\qquad\phi^{(0)}\in D\left(\Delta^{\infty}\right)
\]
such that:

\begin{enumerate}
\item 

The initial variance of the position $x$ and the velocity $\frac{p}{\rho}$ of the tracer particle fulfills
\begin{equation}
\left\langle \chi^{(0)},\left(x-\left\langle \chi^{(0)},x\chi^{(0)}\right\rangle \right)^{2}+\left(\frac{p-\left\langle \chi^{(0)},p\chi^{(0)}\right\rangle }{\rho}\right)^{2},\chi^{(0)}\right\rangle \leq\frac{C}{\rho^{\delta}}\label{eq:initial_spread}
\end{equation}
for a fixed constant $\delta$, with $0<\delta\leq 1$.
\item The support of $\phi^{(0)}$ is contained in a compact region $\Lambda\subset\mathbb{R}^{3}$,
and, as $|\Lambda|\to\infty$, the initial wave function becomes flat,
in the sense that
\begin{equation}
\left\Vert \widehat{\phi^{(0)}}\right\Vert _{1}=C\left|\Lambda\right|{}^{-1/2},\qquad\left\Vert \widehat{\nabla\phi^{(0)}}\right\Vert _{1}\leq C\left|\Lambda\right|{}^{-5/6}.\label{eq:initial_gas_wf}
\end{equation}

\end{enumerate}
\item As initial conditions for the macroscopic dynamics (\ref{eq:macro_gas})-(\ref{eq:macro_tracer})
we choose
\begin{equation}
\epsilon_{t}|_{t=0}=0,\qquad X_{t}|_{t=0}=\left\langle \chi^{(0)},x\,\chi^{(0)}\right\rangle ,\qquad\dot{X}_{t}|_{t=0}=\left\langle \chi^{(0)},\frac{p}{\rho}\,\chi^{(0)}\right\rangle .\label{eq:semi-classical initial values}
\end{equation}
Furthermore, we define
\[
\phi^{(\mathrm{ref})}:\mathbb{R}\to L^{2}(\mathbb{R}^{3},\mathbb C),\qquad t\mapsto\phi_{t}^{(\mathrm{ref})}
\]
 to be the solution to (\ref{eq:free_reference}) with
\begin{equation}
\phi_{t}^{(\mathrm{ref)}}|_{t=0}=\left|\Lambda\right|^{1/2}\phi^{(0)}.\label{eq:initial_reference}
\end{equation}

\end{enumerate}
\end{defn}
\begin{rem}
(i) Note that an example of a function $\chi^{(0)}$ satisfying (\ref{eq:initial_spread})
is a Gaussian wave packet with a variance in the position comprised between
$\rho^{-\gamma}$ and $\rho^{\gamma-1}$, for some $0<\gamma<1$. By
Heisenberg's uncertainty principle, the variance in the momentum is then
between $\rho^{1-\gamma}$ and $\rho^{\gamma}$, and, since the mass of the tracer particle 
is of order $\rho$, the variance of its velocity is between $\rho^{-\gamma}$
and $\rho^{\gamma-1}$. Hence, $\delta$ can be chosen to be $\min\left\{ \gamma,1-\gamma\right\} $.
(ii) A product wave function like (\ref{eq:initial_psi}) is of course
a very special initial condition. However, this condition can be relaxed, as pointed out in Remark \ref{rem:ini-con} below.
\end{rem}
In order to keep our analysis simple, we assume the potentials $V,W$ to be smooth functions of compact support, i.e.,
\begin{equation}
V,W\in{\cal C}_{c}^{\infty}(\mathbb{R}^{3},\mathbb{R}).\label{eq:potentials}
\end{equation}
It is a standard result that the equations of motion (\ref{eq:microscopic}),
(\ref{eq:free_reference}), and (\ref{eq:macro_gas})-(\ref{eq:macro_tracer}),
with initial conditions as chosen above, have smooth solutions.
\begin{defn}
\label{def:solutions}We denote by

\begin{eqnarray}
\Psi:\mathbb{R}\to L^{2}(\mathbb{R}^{3},\mathbb{C})\otimes L^{2}(\mathbb{R}^{3},\mathbb{C})^{\odot N}\subset{\cal H},\qquad t & \mapsto & \Psi_{t}\label{eq:micro psi},
\end{eqnarray}
the unique solution to the Schr\"odinger  equation (\ref{eq:microscopic}) 
with initial condition given by (\ref{eq:initial_psi}); by
\[
\phi^{(\mathrm{ref})}:\mathbb{R}\to L^{2}(\mathbb{R}^{3},\mathbb{C}),\qquad t\mapsto\phi_{t}^{(\mathrm{ref})},
\]
the unique solution to Eq. (\ref{eq:free_reference}) 
with initial data as in (\ref{eq:initial_reference}); and by

\begin{eqnarray}
\epsilon\times X\times \dot X:\mathbb{R}\to L^{2}(\mathbb{R}^{3},\mathbb{C})\times\mathbb{R}^{3}\times\mathbb{R}^{3},\qquad t & \mapsto & (\epsilon_{t},X_{t},\dot{X}_{t})\label{eq:macro solution}
\end{eqnarray}
the unique solution to equations (\ref{eq:macro_gas})-(\ref{eq:macro_tracer})
with initial data as in (\ref{eq:semi-classical initial values}). \end{defn}
\begin{rem}
Note that assumptions (\ref{eq:potentials}) and $\chi^{(0)},\phi^{(0)}\in D\left(\Delta^{\infty}\right)$,
in \defref{initial_values}, are much stronger than necessary. As can
be seen from the norms used in our proofs, the results presented below
hold for a considarably more general class of potentials and initial wave
functions. Finding optimal conditions is, however, not our aim
in this paper.
\end{rem}
The main result in our paper is the following theorem:
\begin{thm}
\label{thm:main}Let \foreignlanguage{english}{\textup{
\[
\rho_{t}^{(\mathrm{micro})}:=q_{t}^{(\mathrm{ref})}\, tr{}_{x,y_{2},\ldots,y_{N}}\bigg(\left|\Lambda\right|\left|\Psi_{t}\right\rangle \left\langle \Psi_{t}\right|\bigg)q_{t}^{(\mathrm{ref})},\qquad\rho_{t}^{(\mathrm{macro})}:=\left|\epsilon_{t}\right\rangle \left\langle \epsilon_{t}\right|
\]
}}denote the density matrices of the gas excitations w.r.t. $\phi_{t}^{(\mathrm{ref})}$ of the microscopic
and the macroscopic descriptions, respectively; for the definition of $q_{t}^{(\mathrm{ref})}$, see (\ref{eq:qtref}).
There exist $C_{1},C_{2}\in\mathtt{Bounds}$ such that, for all $t\in\mathbb{R}$
and sufficiently large $\rho$ and $\left|\Lambda\right|$, the following
estimates hold true:
\begin{enumerate}[label=(\roman*)]
\item 
  \begin{align}
    \left\Vert \rho_{t}^{(\mathrm{micro})}-\rho_{t}^{(\mathrm{macro})}\right\Vert &\leq C_{1}(t)\left(\alpha_{0}^{1/2}+\rho^{-1/2}+\left|\Lambda\right|^{-1/3}\right) \label{eq:density-ieq}\\
    &\leq C_{1}(t)\left(\rho^{-\frac{1}{2}\min\{1,\delta\}}+\left|\Lambda\right|^{-1/3}\right).\label{eq:density-ieq2}
    \end{align}
\selectlanguage{english}%
\item 
\begin{align}
\|X_{t}-\left\langle x\right\rangle _{t}\|+\|\dot{X}_{t}-\left\langle \frac{p}{\rho}\right\rangle _{t}\|&\leq C_{2}(t)\left(\alpha_{0}^{1/2}+\rho^{-1/2}+\left|\Lambda\right|^{-1/3}\right) \label{eq:traj-ieq}\\
&\leq C_{2}(t)\left(\rho^{-\frac{1}{2}\min\{1,\delta\}}+\left|\Lambda\right|^{-1/3}\right).\label{eq:traj-ieq2}
  \end{align}
\end{enumerate}
Here $\alpha_{0}$, defined in (\ref{eq:alpha_def}) below, reflects the dependence on the initial condition (\ref{eq:initial_psi}).
\end{thm}
The operator $tr{}_{x,y_{2},\ldots,y_{N}}$ stands for tracing
out the degrees of freedom specified in the subscript. Estimate
(i) in \thmref{main} quantifies how well the density matrix of the
gas excitations w.r.t. the reference state $\phi_{t}^{(\mathrm{ref})}$ is
approximated by the effective density matrix $\left|\epsilon_{t}\right\rangle \left\langle \epsilon_{t}\right|$,
while estimate (ii) quantifies how close the expected position and
velocity of the tracer particle are to the classical ones. Hence,
for large $\rho$ and $\left|\Lambda\right|$, the microscopic
dynamics and macroscopic dynamics yield arbitrarily close predictions, and, for
practical purposes, one may thus use the macroscopic equations
to study the behavior of the system.

\begin{rem}\label{rem:ini-con}
(i) We remark that inequalities (\ref{eq:density-ieq}) and (\ref{eq:traj-ieq}) hold for any initial conditions such that
\[
\alpha_{0}= o_{|\rho|\to\infty}(1)+o_{|\Lambda|\to\infty}(1),
\]
while inequalities (\ref{eq:density-ieq2}) and (\ref{eq:traj-ieq2}) only hold for the initial conditions (\ref{eq:initial_psi}) satisfying (\ref{eq:initial_spread}) and (\ref{eq:initial_gas_wf}). (ii) Theorem \ref{thm:main} can easily be generalized to systems of $M>1$ interacting tracer particles, as discussed in Remark \ref{rem:many-tracer} below.
\end{rem}

\section{Proof of Main Theorem}

The strategy of the proof is a two step procedure. First, we probe
how well the gas particles in $\Psi_{t}$ retain the product structure
encoded in the initial wave function (\ref{eq:initial_psi}). In particular,
we compare the microscopic dynamics to the effective dynamics generated
by
\begin{eqnarray}
i\partial_{t}\varphi_{t}(y) & = & \left(-\Delta_{y}+W\left(\left\langle x\right\rangle _{t}-y\right)\right)\varphi_{t}(y)\label{eq:intermediate}
\end{eqnarray}
for the initial value
\begin{equation}
\varphi_{t}|_{t=0}=\phi^{(0)}.\label{eq:intermediate initial value}
\end{equation}

\begin{defn}
We denote by
\[
 \varphi:\mathbb{R}\to L^{2}(\mathbb{R}^{3},\mathbb{C}),\qquad t\mapsto\varphi_{t}
\]
the unique solution to (\ref{eq:intermediate}) with initial condition
(\ref{eq:intermediate initial value}).
\end{defn}
As shown in \cite{Pickl_2011_A-Simple-Derivation-of-Mean-Field-Limits-for-Quantum-Systems}
for pure Bose gases without a tracer particle, it is convenient
to control the deviation of the gas wave function from a product wave function with the
help of a Grönwall-type estimate of the form
\begin{equation}
\frac{d}{dt}\widetilde{\alpha}_{t}\leq C\,\widetilde{\alpha}_{t}+\frac{C}{N},\qquad\text{where}\qquad \widetilde{\alpha}_{t}:=\left\langle q_{1}^{\varphi_{t}}\right\rangle _{t}.\label{eq:mean-field alpha}
\end{equation}
The quantity $\widetilde{\alpha}_{t}$ counts the relative number
of tensor components in $\Psi_{t}$ not showing product structure.
This can be seen best by means of the identity (5) in Lemma 2.2 of \cite{Pickl_2010}, i.e.,
\begin{eqnarray*}
\left\langle q_{1}^{\varphi_{t}}\right\rangle _{t}=\sum_{k=0}^{N}\frac{k}{N}\left\langle q_{1}^{\varphi_{t}}\odot q_{2}^{\varphi_{t}}\odot\ldots\odot q_{k}^{\varphi_{t}}\odot p_{k+1}^{\varphi_{t}}\odot p_{k+2}^{\varphi_{t}}\odot\ldots\odot p_{N}^{\varphi_{t}}\right\rangle _{t}
\end{eqnarray*}
where $\odot$ denotes the symmetrized tensor product (see (\ref{eq:projectors}) for the definition of $p_{k}^{\varphi_{t}}$ and $q_{k}^{\varphi_{t}}$). Hence, $\widetilde{\alpha}_{t}$ corresponds to the expectation value of the ratio $k/N$ between the number of particles, $k$, that are not in the state of the Hartree wave function and the total number of gas particles $N$.

The situation considered in this paper is more complicated because
of the presence of the tracer particle, which couples to the gas and
 generates entanglement between its state and the state of the
gas particles. As a consequence, the error estimates
will not only depend on $\widetilde{\alpha}_{t}$ but also on the
position and momentum distribution of the tracer particle wave function.
To see this we consider the example of an initial wave function $\chi^{(0)}$
of the tracer particle formed by a superposition of two wave packets
whose supports are separated by a distance of order one. In the worst
case the mean position $\left\langle x\right\rangle _{t}$ could then
be somewhere in between the supports of these wave packets. In a situation like this, the effective interaction term $W\left(\left\langle x\right\rangle _{t}-y\right)$
in (\ref{eq:intermediate}) has nothing to do with the actual interaction
given by $\sum_{k=1}^{N}W(x-y_{k})$ in (\ref{eq:micro hamiltonian}),
and there is no reason to expect that $\widetilde{\alpha}_{t}$ stays small.

Moreover, as discussed above, $\widetilde{\alpha}_{t}$ is the expected
ratio $k/N$ of ``bad'' gas particles w.r.t. to $N$. Yet, for the control of the dynamics of the tracer particle, this ratio will not be very relevant, because, due to the support of $W$, the tracer particle only sees $O(\rho)$ many gas particles at a time. In the worst case scenario however, even though $k/N$ might be small, all the ``bad'' gas particles could actually be in the vicinity of the tracer particle. It is therefore important to know how many gas particles are ``bad'', as
compared to a number of gas particles of $O(\rho)$. The latter amounts to estimating the quantity $|\Lambda|\left\langle q_{1}^{\varphi_{t}}\right\rangle_{t}$ that gives the expected ratio $k/\rho$.

We must  therefore carefully adapt $\widetilde{\alpha}_{t}$ to our situation. It turns out that, among appropriate choices that make the desired estimates
quite easy, the following one is convenient.
\begin{defn}
We define 
\begin{equation}
\alpha_{t}:=\sqrt{\left\langle \left(x-\left\langle x\right\rangle _{t}\right)^{2}\right\rangle _{t}^{2}+\left\langle \left(\frac{p-\left\langle p\right\rangle _{t}}{\rho}\right)^{2}\right\rangle _{t}^{2}+\left(\left|\Lambda\right|\left\langle q_{1}^{\varphi_{t}}\right\rangle _{t}\right)^{2}+\left(\left|\Lambda\right|^{2}\left\langle q_{1}^{\varphi_{t}}q_{2}^{\varphi_{t}}\right\rangle _{t}\right)^{2}}\label{eq:alpha_def}
\end{equation}
for $t\in\mathbb{R}$.
\end{defn}
Note that the function $t\mapsto\alpha_{t}$ is smooth. A key estimate is the following lemma.
\begin{lem}
\label{lem:alpha}There are $C_{\alpha}^{(1)},C_{\alpha}^{(2)}\in\mathtt{Bounds}$
such that, for all $t\geq0$, the following estimate holds:
\[
\frac{d}{dt}\alpha_{t}\leq C_{\alpha}^{(1)}(t)\alpha_{t}+C_{\alpha}^{(2)}(t)\rho^{-1}.
\]

\end{lem}
In a second step we then control the error made when replacing the
mean position $t\mapsto\left\langle x\right\rangle _{t}$ in (\ref{eq:intermediate}), which fulfills the Ehrenfest equation
\begin{equation}
\frac{d^{2}}{dt^{2}}\left\langle x\right\rangle _{t}=\left\langle -\nabla V(x)-\frac{1}{\rho}\nabla\sum_{k=1}^{N}W(x-y_{k})\right\rangle _{t},\label{eq:ehrenfest}
\end{equation}
by the classical trajectory $t\mapsto X_{t}$ obeying (\ref{eq:macro_tracer}).
Furthermore, in order to probe the excited modes of the gas, we need
good control on how well the effective wave function of the
gas, $\left|\Lambda\right|^{1/2}\varphi_{t}$, approximates 
\[
\phi_{t}^{(\mathrm{ref})}+\epsilon_{t}.
\]

For this second step, too, we will invoke a Grönwall-type estimate. We will consider the following expression.
\begin{defn}
We define 
\begin{equation}
t\mapsto\beta_{t}:=\sqrt{\left(X_{t}-\left\langle x\right\rangle _{t}\right)^{2}+\left(\frac{d\left(X_{t}-\left\langle x\right\rangle _{t}\right)}{dt}\right)^{2}+\left\Vert \left(\phi_{t}^{(\mathrm{ref})}+\epsilon_{t}\right)-\left|\Lambda\right|^{1/2}\varphi_{t}\right\Vert _{2}^{2}},\label{eq:beta_def}
\end{equation}
for $t\in\mathbb{R}$. 
\end{defn}
Note that the function $t\mapsto\beta_{t}$ is smooth. We will prove the following estimate.
\begin{lem}
\label{lem:beta}There are $C_{\beta}^{(1)},C_{\beta}^{(2)}\in\mathtt{Bounds}$
such that, for all $t\geq0$, the following estimate holds:
\begin{equation}
\frac{d}{dt}\beta_{t}\leq C_{\beta}^{(1)}(t)\left(\beta_{t}+\beta_{t}^{2}\right)+C_{\beta}^{(2)}(t)\left(\sqrt{\alpha_{t}}+\alpha_{t}+\left|\Lambda\right|^{-1/3}\right).\label{eq:beta estimate}
\end{equation}

\end{lem}
This will complete the second step, and combination of Lemma \ref{lem:alpha} and Lemma \ref{lem:beta} will complete the proof of our main theorem.
\begin{proof}[\textbf{Proof of \thmref{main}}]
With the help of Gr\"onwall's Lemma, \lemref{alpha} implies that there is some $C_{\alpha}\in\mathtt{Bounds}$
such that
\begin{equation}
\alpha_{t}\leq C_{\alpha}(t)\left(\alpha_{0}+\rho^{-1}\right).\label{eq:alpha bound}
\end{equation}
Moreover, the function $\beta_{t}$ is smooth, and, by our choice of initial conditions
\[\beta_{t}|_{t=0}=0.\]
We assume further that
\begin{eqnarray}
\alpha_{t}|_{t=0}=o_{|\rho|\to\infty}(1)+o_{|\Lambda|\to\infty}(1),\label{eq:beta-cond}
\end{eqnarray}
i.e., that the right-hand side of (\ref{eq:beta-cond}) becomes arbitrarily small for sufficiently large parameters $\rho$ and $|\Lambda|$. Hence, there exists a time $T>0$ such that, for all $0\leq t < T$, the bound
$0\leq \beta_{t} \leq 1$ holds. The estimate in Lemma \ref{lem:beta} then
implies 
\begin{eqnarray}
\frac{d}{dt}\beta_{t} & \leq & 2C_{\beta}^{(1)}(t)\beta_{t}+C_{\beta}^{(2)}(t)\left(\sqrt{C_{\alpha}(t)\left(\alpha_{0}+\rho^{-1}\right)}+C_{\alpha}(t)\left(\alpha_{0}+\rho^{-1}\right)+\left|\Lambda\right|^{-1/3}\right),\label{eq:beta ineq}
\end{eqnarray}
so that Grönwall's Lemma guarantees the existence of some $C_{\beta}\in\mathtt{Bounds}$ such that
\begin{equation}
\beta_{t}\leq C_{\beta}(t)\left(\alpha_{0}^{1/2}+\rho^{-1/2}+\left|\Lambda\right|^{-1/3}\right).\label{eq:beta bound}
\end{equation}
Let $\overline{T}(\rho,\left|\Lambda\right|)$ be the supremum of all such times $T$, and let us assume that $\overline{T}(\rho,\left|\Lambda\right|)$
is uniformly bounded. Then, upon choosing $\rho$ and $\left|\Lambda\right|$
large enough in (\ref{eq:beta bound}), we can arrange for $\beta_{\overline{T}(\rho,\left|\Lambda\right|)}<\frac{1}{2}$.
However, from Grönwall's Lemma it follows that  (\ref{eq:beta bound})
holds for some time $t>\overline{T}(\rho,\left|\Lambda\right|)$. This contradicts our assumption, and hence
\[
\lim_{\rho,|\Lambda|\to\infty}\overline{T}(\rho,|\Lambda|)=\infty.
\]
We therefore conclude that (\ref{eq:beta bound}) holds for arbitrarily large $t\geq 0$, provided $\rho$ and $|\Lambda|$ are large enough. This proves claim (ii) of Theorem \thmref{main}.

Next, we prove claim (i). Similarly to \cite{Pickl_2011_A-Simple-Derivation-of-Mean-Field-Limits-for-Quantum-Systems},
we start by inserting the identity $\mathbbm{1}=p_{1}^{\varphi_{t}}+q_{1}^{\varphi_{t}}$
on the left and right-hand side of $\left|\Psi_{t}\right\rangle \left\langle \Psi_{t}\right|$,
which yields
\begin{eqnarray}
\left\Vert \rho_{t}^{(\mathrm{micro})}-\rho_{t}^{(\mathrm{macro})}\right\Vert  & \equiv & \left\Vert \left|\Lambda\right|q_{t}^{(\mathrm{ref})}\, tr{}_{x,y_{2},\ldots,y_{N}}\left|\Psi_{t}\right\rangle \left\langle \Psi_{t}\right|q_{t}^{(\mathrm{ref})}-\left|\epsilon_{t}\right\rangle \left\langle \epsilon_{t}\right|\right\Vert \nonumber \\
 & \leq & \left\Vert \left|\Lambda\right|q_{t}^{(\mathrm{ref})}\, tr{}_{x,y_{2},\ldots,y_{N}}\left[p_{1}^{\varphi_{t}}\left|\Psi_{t}\right\rangle \left\langle \Psi_{t}\right|p_{1}^{\varphi_{t}}\right]q_{t}^{(\mathrm{ref})}-\left|\epsilon_{t}\right\rangle \left\langle \epsilon_{t}\right|\right\Vert \label{eq:density_PP}\\
 &  & +2\left|\Lambda\right|\left\Vert q_{t}^{(\mathrm{ref})}\, tr{}_{x,y_{2},\ldots,y_{N}}\left[p_{1}^{\varphi_{t}}\left|\Psi_{t}\right\rangle \left\langle \Psi_{t}\right|q_{1}^{\varphi_{t}}\right]q_{t}^{(\mathrm{ref})}\right\Vert \label{eq:density_PQ}\\
 &  & +\left|\Lambda\right|\left\Vert q_{t}^{(\mathrm{ref})}\, tr{}_{x,y_{2},\ldots,y_{N}}\left[q_{1}^{\varphi_{t}}\left|\Psi_{t}\right\rangle \left\langle \Psi_{t}\right|q_{1}^{\varphi_{t}}\right]q_{t}^{(\mathrm{ref})}\right\Vert .\label{eq:density_QQ}
\end{eqnarray}
In order to estimate (\ref{eq:density_PP}) we need the preliminary
bound 
\begin{align}
 & \left\Vert \left|\Lambda\right|q_{t}^{(\mathrm{ref})}\left|\varphi_{t}\right\rangle \left\langle \varphi_{t}\right|q_{t}^{(\mathrm{ref})}-\left|\epsilon_{t}\right\rangle \left\langle \epsilon_{t}\right|\right\Vert \nonumber \\
= & \left\Vert q_{t}^{(\mathrm{ref})}\left|\left|\Lambda\right|^{1/2}\varphi_{t}-\left(\phi_{t}^{(\mathrm{ref})}+\epsilon_{t}\right)+\left(\phi_{t}^{(\mathrm{ref})}+\epsilon_{t}\right)\right\rangle \left\langle \left|\Lambda\right|^{1/2}\varphi_{t}-\left(\phi_{t}^{(\mathrm{ref})}+\epsilon_{t}\right)+\left(\phi_{t}^{(\mathrm{ref})}+\epsilon_{t}\right)\right|q_{t}^{(\mathrm{ref})}-\left|\epsilon_{t}\right\rangle \left\langle \epsilon_{t}\right|\right\Vert \nonumber \\
\leq & \left\Vert \left|\Lambda\right|^{1/2}\varphi_{t}-\left(\phi_{t}^{(\mathrm{ref})}+\epsilon_{t}\right)\right\Vert _{2}^{2}+2\left\Vert \left|\Lambda\right|^{1/2}\varphi_{t}-\left(\phi_{t}^{(\mathrm{ref})}+\epsilon_{t}\right)\right\Vert _{2}\left\Vert \epsilon_{t}\right\Vert _{2}+\left\Vert q_{t}^{(\mathrm{ref})}\left|\epsilon_{t}\right\rangle \left\langle \epsilon_{t}\right|q_{t}^{(\mathrm{ref})}-\left|\epsilon_{t}\right\rangle \left\langle \epsilon_{t}\right|\right\Vert \nonumber \\
\leq & \beta_{t}^{2}+2\beta_{t}+\left\Vert p_{t}^{(\mathrm{ref})}\left|\epsilon_{t}\right\rangle \left\langle \epsilon_{t}\right|p_{t}^{(\mathrm{ref})}\right\Vert +2\left\Vert p_{t}^{(\mathrm{ref})}\left|\epsilon_{t}\right\rangle \left\langle \epsilon_{t}\right|\right\Vert \nonumber \\
\leq & \beta_{t}^{2}+2\beta_{t}+\frac{\widetilde{C}_{\epsilon}(t)^{2}}{\left|\Lambda\right|}+\frac{\widetilde{C}_{\epsilon}(t)}{\left|\Lambda\right|^{1/2}},\label{eq:phi-eps-dist}
\end{align}
where in the last two lines we are using (\ref{eq:beta_def}) and \lemref{ref_and_excitation}.
Next, we establish estimates on the terms (\ref{eq:density_PP}),(\ref{eq:density_PQ}),(\ref{eq:density_QQ}):\\

\noindent\noun{Term (\ref{eq:density_PP})}: By Fubini, one finds that
\[
\left\langle \varphi_{t}\right|tr{}_{x,y_{2},\ldots,y_{N}}\left[\left|\Psi_{t}\right\rangle \left\langle \Psi_{t}\right|\right]\left|\varphi_{t}\right\rangle =\left\langle \Psi_{t},p_{1}^{\varphi_{t}}\Psi_{t}\right\rangle =1-\left\langle \Psi_{t},q_{1}^{\varphi_{t}}\Psi_{t}\right\rangle 
\]
so that
\begin{eqnarray}
(\ref{eq:density_PP}) & = & \left\Vert \left|\Lambda\right|q_{t}^{(\mathrm{ref})}\left|\varphi_{t}\right\rangle \left\langle \varphi_{t}\right|tr{}_{x,y_{2},\ldots,y_{N}}\left[\left|\Psi_{t}\right\rangle \left\langle \Psi_{t}\right|\right]\left|\varphi_{t}\right\rangle \left\langle \varphi_{t}\right|q_{t}^{(\mathrm{ref})}-\left|\epsilon_{t}\right\rangle \left\langle \epsilon_{t}\right|\right\Vert \nonumber \\
 & = & \left\Vert \left(1-\left\langle \Psi_{t},q_{1}^{\varphi_{t}}\Psi_{t}\right\rangle \right)\left[\left|\Lambda\right|q_{t}^{(\mathrm{ref})}\left|\varphi_{t}\right\rangle \left\langle \varphi_{t}\right|q_{t}^{(\mathrm{ref})}-\left|\epsilon_{t}\right\rangle \left\langle \epsilon_{t}\right|\right]+\left\langle \Psi_{t},q_{1}^{\varphi_{t}}\Psi_{t}\right\rangle \left|\epsilon_{t}\right\rangle \left\langle \epsilon_{t}\right|\right\Vert \nonumber \\
 & \leq & \left\Vert \left(1-\left\langle \Psi_{t},q_{1}^{\varphi_{t}}\Psi_{t}\right\rangle \right)\left[\left|\Lambda\right|q_{t}^{(\mathrm{ref})}\left|\varphi_{t}\right\rangle \left\langle \varphi_{t}\right|q_{t}^{(\mathrm{ref})}-\left|\epsilon_{t}\right\rangle \left\langle \epsilon_{t}\right|\right]\right\Vert +\left|\left\langle \Psi_{t},q_{1}^{\varphi_{t}}\Psi_{t}\right\rangle \right|\left\Vert \epsilon_{t}\right\Vert _{2}^{2}\nonumber \\
 & \leq & 2\left(\beta_{t}^{2}+2\beta_{t}+\frac{\widetilde{C}_{\epsilon}(t)^{2}}{\left|\Lambda\right|}+\frac{\widetilde{C}_{\epsilon}(t)}{\left|\Lambda\right|^{1/2}}\right)+\frac{\alpha_{t}}{\left|\Lambda\right|}C_{\epsilon}(t)^{2},\label{eq:density_PP_est}
\end{eqnarray}
where we have used inequality (\ref{eq:alpha bound}) to get
\[
\left|1-\left\langle \Psi_{t},q_{1}^{\varphi_{t}}\Psi_{t}\right\rangle \right|\leq1+\frac{\alpha_{t}}{\left|\Lambda\right|}\leq2
\]
for $\left|\Lambda\right|,\rho\gg1$, and, furthermore, inequality (\ref{eq:phi-eps-dist}),
definition (\ref{eq:alpha_def}), and \lemref{ref_and_excitation}.\\

\noindent\noun{Term (\ref{eq:density_PQ}):}
\begin{eqnarray}
(\ref{eq:density_PQ}) & = & 2\left|\Lambda\right|\left\Vert q_{t}^{(\mathrm{ref})}\, tr{}_{x,y_{2},\ldots,y_{N}}\left[p_{1}^{\varphi_{t}}\left|\Psi_{t}\right\rangle \left\langle \Psi_{t}\right|q_{1}^{\varphi_{t}}\right]q_{t}^{(\mathrm{ref})}\right\Vert \nonumber \\
 & \leq & 2\left|\Lambda\right|\left\Vert q_{t}^{(\mathrm{ref})}\varphi_{t}\right\Vert \left\Vert q_{1}^{\varphi_{t}}\Psi_{t}\right\Vert \nonumber \\
 & = & 2\left|\Lambda\right|\left\Vert q_{t}^{(\mathrm{ref})}\left(\frac{\phi^{(\mathrm{ref})}+\epsilon_{t}}{\left|\Lambda\right|^{1/2}}+\varphi_{t}-\frac{\phi^{(\mathrm{ref})}+\epsilon_{t}}{\left|\Lambda\right|^{1/2}}\right)\right\Vert \left\Vert q_{1}^{\varphi_{t}}\Psi_{t}\right\Vert \nonumber \\
 & \leq & 2\left|\Lambda\right|\left(\left\Vert \frac{\epsilon_{t}}{\left|\Lambda\right|^{1/2}}\right\Vert +\left\Vert \varphi_{t}-\frac{\phi^{(\mathrm{ref})}+\epsilon_{t}}{\left|\Lambda\right|^{1/2}}\right\Vert \right)\sqrt{\frac{\alpha_{t}}{\left|\Lambda\right|}}\nonumber \\
 & \leq & 2\sqrt{\alpha_{t}}\left(C_{\epsilon}(t)+\beta_{t}\right)\label{eq:density_PQ_est}
\end{eqnarray}
where we have used definitions (\ref{eq:alpha_def}), (\ref{eq:beta_def})
and \lemref{ref_and_excitation}.\\

\noindent\noun{Term (\ref{eq:density_QQ}): }
\begin{eqnarray}
(\ref{eq:density_QQ}) & = & \left|\Lambda\right|\left\Vert q_{t}^{(\mathrm{ref})}\, tr{}_{x,y_{2},\ldots,y_{N}}\left[q_{1}^{\varphi_{t}}\left|\Psi_{t}\right\rangle \left\langle \Psi_{t}\right|q_{1}^{\varphi_{t}}\right]q_{t}^{(\mathrm{ref})}\right\Vert \nonumber \\
 & \leq & \left|\Lambda\right|\left\Vert q_{1}^{\varphi_{t}}\Psi_{t}\right\Vert ^{2}\nonumber \\
 & \leq & \alpha_{t},\label{eq:density_QQ_est}
\end{eqnarray}
where we have used definition (\ref{eq:alpha_def}).\\

Collecting estimates (\ref{eq:density_PP_est}), (\ref{eq:density_PQ_est}),
(\ref{eq:density_QQ_est}), and using (\ref{eq:alpha bound}) as well
as (\ref{eq:beta bound}), we find that
\begin{eqnarray}
\left\Vert \rho_{t}^{(\mathrm{micro})}-\rho_{t}^{(\mathrm{macro})}\right\Vert  & \leq & \beta_{t}^{2}+2\beta_{t}+\frac{\widetilde{C}_{\epsilon}(t)^{2}}{\left|\Lambda\right|}+\frac{\widetilde{C}_{\epsilon}(t)}{\left|\Lambda\right|^{1/2}}+\frac{\alpha_{t}}{\left|\Lambda\right|}C_{\epsilon}(t)^{2}+2\sqrt{\alpha_{t}}\left(C_{\epsilon}(t)+\beta_{t}\right)+\alpha_{t}\nonumber\\
 & \leq & C_{1}(t)\left(\alpha_{0}^{1/2}+\rho^{-1/2}+\left|\Lambda\right|^{-1/3}\right),\label{eq:rho-est}
\end{eqnarray}
for some $C_{1}\in\mathtt{Bounds}$, as well as
\begin{align}
\|X_{t}-\left\langle x\right\rangle _{t}\|+\|\dot{X}_{t}-\left\langle \frac{p}{\rho}\right\rangle _{t}\| &\leq \beta_{t}\label{eq:est-by-beta}\\
&\leq C_{\beta}(t)\left(\alpha_{0}^{1/2}+\rho^{-1/2}+\left|\Lambda\right|^{-1/3}\right),\label{eq:spread-est}
\end{align}
where inequality (\ref{eq:est-by-beta}) follows from the definition of the function $\beta_{t}$ given in (\ref{eq:beta_def}). The choice of initial conditions (\ref{eq:initial_psi}), satisfying conditions (\ref{eq:initial_spread}) and (\ref{eq:initial_gas_wf}), ensures that
\begin{eqnarray}
\alpha_{t}|_{t=0}\leq\frac{C}{\rho^{\delta}}.\label{eq:init-alpha-beta}
\end{eqnarray}

Hence, (\ref{eq:init-alpha-beta}) together with inequality (\ref{eq:rho-est}) proves claim (i), and, together with inequality (\ref{eq:spread-est}), proves claim (ii).
\end{proof}

In the rest of this section we present proofs of Lemma \ref{lem:alpha}
and Lemma \ref{lem:beta} which provide the estimates on the time derivatives of the quantities $\alpha_{t}$ and
$\beta_{t}$ defined in (\ref{eq:alpha_def}) and (\ref{eq:beta_def}).

The aim in each of these estimates is to show that either the corresponding terms can be bounded in terms of $\alpha_{t}$ and $\beta_{t}$ or that they are small if one of the parameters $|\Lambda|$ or $\rho$ is large. The main mechanisms exploited in our strategy are:
\begin{itemize}
  \item Expansion of differences for smooth functions $f$ of the form
    \[f(x)-f\left(\left\langle x\right\rangle_{t}\right)=R^{f}(x,\left\langle x\right\rangle_{t})Ê\cdot\left(x-\left\langle x\right\rangle_{t}\right)\] where $R^{f}$ denotes Taylor's remainder term. This is meant to take advantage of the estimate
    \[\left\|\left(x-\left\langle x\right\rangle_{t}\right)\Psi_{t}\right\|_{2} \leq \sqrt{\alpha_{t}}.\]
  \item Rearrangement of the arguments of scalar products, if necessary by inserting the identity $\mathbbm{1}=p_{1}^{\varphi_{t}}+q_{1}^{\varphi_{t}}$, such that the operator $q_{1}^{\varphi_{t}}$ acts directly on $\Psi_{t}$. This is meant to take advantage of the estimate
    \[\left\|q_{1}^{\varphi_{t}}\Psi_{t}\right\|_{2} \leq \frac{\sqrt{\alpha_{t}}}{|\Lambda|^{1/2}}.\]
  
  \item Whenever the Hartree wave function $\varphi_{t}$ is integrated against a function $f$ with support contained  in a volume of $O(1)$ one gains a factor $|\Lambda|^{-1/2}$, by Lemma \ref{lem:proagation}. Note that, for example,
  \[ p_{1}^{\varphi_{t}}f(x-y_{1})p_{1}^{\varphi_{t}}=(f*|\varphi_{t}|^{2})(x)p_{1}^{\varphi_{t}},\]  
  and therefore, thanks to the argument above, one gains a factor $|\Lambda|^{-1}$.  
  
\end{itemize}

\begin{proof}[\textbf{Proof of \lemref{alpha}}]
Since $\alpha_{t}$ is a smooth function of $t$, we can estimate its derivative
by
\begin{eqnarray}
\frac{d}{dt}\alpha_{t} & \leq & \left|\frac{d}{dt}\left\langle \left(x-\left\langle x\right\rangle _{t}\right)^{2}\right\rangle _{t}\right|\label{eq:x_variance}\\
 &  & +\left|\frac{d}{dt}\left\langle \left(\frac{p-\left\langle p\right\rangle _{t}}{\rho}\right)^{2}\right\rangle _{t}\right|\label{eq:p_variance}\\
 &  & +\left|\Lambda\right|\left|\frac{d}{dt}\left\langle q_{1}^{\varphi_{t}}\right\rangle _{t}\right|\label{eq:1_q}\\
 &  & +\left|\Lambda\right|^{2}\left|\frac{d}{dt}\left\langle q_{1}^{\varphi_{t}}q_{2}^{\varphi_{t}}\right\rangle _{t}\right|.\label{eq:2_q}
\end{eqnarray}
Denoting commutators by $\left[\cdot,\cdot\right]$ and anti-commutators
by $\left\{ \cdot,\cdot\right\} $, we shall use the following auxiliary
computation in the estimates of the individual terms (\ref{eq:p_variance})-(\ref{eq:2_q}):
Let $A$ be an arbitary self-adjoint operator; then the estimate 
\begin{eqnarray}
\left|\frac{d}{dt}\left\langle \left(A-\left\langle A\right\rangle _{t}\right)^{2}\right\rangle _{t}\right| & = & \left|i\left\langle \left[H,\left(A-\left\langle A\right\rangle _{t}\right)^{2}\right]\right\rangle +\left\langle \frac{d}{dt}\left(A-\left\langle A\right\rangle _{t}\right)^{2}\right\rangle _{t}\right|\nonumber \\
 & = & \left|i\left\langle \left\{ \left[H,A-\left\langle A\right\rangle _{t}\right],A-\left\langle A\right\rangle _{t}\right\} \right\rangle _{t}-2\left\langle A-\left\langle A\right\rangle _{t}\right\rangle _{t}\frac{d}{dt}\left\langle A\right\rangle _{t}\right|\nonumber \\
 & = & \left|\left\langle \left\{ \left[H,A\right],A-\left\langle A\right\rangle _{t}\right\} \right\rangle _{t}\right|\nonumber \\
 & \leq & 2\left|\left\langle \left[H,A\right]\left(A-\left\langle A\right\rangle _{t}\right)\right\rangle _{t}\right|\label{eq:helper}
\end{eqnarray}
holds true, supposing the expressions on the right-hand side are well defined (recall that $\Psi_{t}$ is normalized).\\

\noindent\noun{term }(\ref{eq:x_variance})\noun{:} Using definition
(\ref{eq:alpha_def}) we estimate
\begin{eqnarray}
(\ref{eq:x_variance}) & = & \left|\frac{d}{dt}\left\langle \left(x-\left\langle x\right\rangle _{t}\right)^{2}\right\rangle _{t}\right|\nonumber \\
 & \leq & 2\left|\left\langle \frac{p-\left\langle p\right\rangle _{t}}{\rho}\cdot\left(x-\left\langle x\right\rangle _{t}\right)\right\rangle _{t}\right|\nonumber \\
 & \leq & 2\sqrt{\left\langle \left(\frac{p-\left\langle p\right\rangle _{t}}{\rho}\right)^{2}\right\rangle _{t}\left|\left\langle \left(x-\left\langle x\right\rangle _{t}\right)^{2}\right\rangle _{t}\right|}\nonumber \\
 & \leq & C\alpha_{t}.\label{eq:q_var_final}
\end{eqnarray}
\noindent\noun{term }(\ref{eq:p_variance})\noun{: }With the help
of \noun{(\ref{eq:helper}) }we get\textbf{
\begin{eqnarray}
(\ref{eq:p_variance}) & = & \left|\frac{d}{dt}\left\langle \left(\frac{p-\left\langle p\right\rangle _{t}}{\rho}\right)^{2}\right\rangle _{t}\right|\nonumber \\
 & \leq & 2\left|\left\langle \left[H,\frac{p}{\rho}\right]\frac{p-\left\langle p\right\rangle _{t}}{\rho}\right\rangle _{t}\right|\nonumber \\
 & = & 2\left|\left\langle \left[\rho V(x)+NW(x-y_{1}),\frac{p}{\rho}\right]\frac{p-\left\langle p\right\rangle _{t}}{\rho}\right\rangle _{t}\right|\nonumber \\
 & \leq & 2\left|\left\langle \nabla V(x)\cdot\frac{p-\left\langle p\right\rangle _{t}}{\rho}\right\rangle _{t}\right|\label{eq:grad_V}\\
 &  & +2\left|\Lambda\right|\left|\left\langle \nabla W(x-y_{1})\cdot\frac{p-\left\langle p\right\rangle _{t}}{\rho}\right\rangle _{t}\right|.\label{eq:grad_W}
\end{eqnarray}
}We expand $\nabla V(x)$ according to
\[
\nabla V(x)=\nabla V(\left\langle x\right\rangle _{t})+R^{\nabla V}(x,\left\langle x\right\rangle _{t})\,(x-\left\langle x\right\rangle _{t}),
\]
where $R^{\nabla V}$ denotes Taylor's remainder term, and, using (\ref{eq:alpha_def}),
we obtain
\begin{eqnarray}
(\ref{eq:grad_V}) & \leq & 2\left|\nabla V(\left\langle x\right\rangle _{t})\cdot\left\langle \frac{p-\left\langle p\right\rangle _{t}}{\rho}\right\rangle _{t}\right|+2\left|\left\langle \left(R^{\nabla V}(x,\left\langle x\right\rangle _{t})\,(x-\left\langle x\right\rangle _{t})\right)\cdot\frac{p-\left\langle p\right\rangle _{t}}{\rho}\right\rangle _{t}\right|\nonumber \\
 & \leq & 2\left\Vert R^{\nabla V}\right\Vert _{\infty}\left\Vert \left(x-\left\langle x\right\rangle _{t}\right)\Psi_{t}\right\Vert _{2}\left\Vert \frac{p-\left\langle p\right\rangle _{t}}{\rho}\Psi_{t}\right\Vert _{2}\nonumber \\
 & \leq & C\alpha_{t}\label{eq:p_var_1}
\end{eqnarray}
The estimate on term (\ref{eq:grad_W}), which depends on $W$, is
more involved. It is convenient to split it as follows, inserting the identity $\mathbbm{1}=p_{1}^{\varphi_{t}}+q_{1}^{\varphi_{t}}$,
\begin{eqnarray}
(\ref{eq:grad_W}) & \leq & 2\left|\Lambda\right|\left|\left\langle p_{1}^{\varphi_{t}}\nabla W(x-y_{1})p_{1}^{\varphi_{t}}\cdot\frac{p-\left\langle p\right\rangle _{t}}{\rho}\right\rangle _{t}\right|\label{eq:pp}\\
 &  & +2\left|\Lambda\right|\left|\left\langle \left(p_{1}^{\varphi_{t}}+q_{1}^{\varphi_{t}}\right)\nabla W(x-y_{1})q_{1}^{\varphi_{t}}\cdot\frac{p-\left\langle p\right\rangle _{t}}{\rho}\right\rangle _{t}\right|\label{eq:p+q_q}\\
 &  & +2\left|\Lambda\right|\left|\left\langle q_{1}^{\varphi_{t}}\nabla W(x-y_{1})p_{1}^{\varphi_{t}}\cdot\frac{p-\left\langle p\right\rangle _{t}}{\rho}\right\rangle _{t}\right|.\label{eq:q_p}
\end{eqnarray}

In order to treat the term (\ref{eq:pp}), we expand $\nabla W$ according
\[
\nabla W(x-y)=\nabla W(\left\langle x\right\rangle _{t}-y)+R^{\nabla W}(x,y,\left\langle x\right\rangle _{t})\,(x-\left\langle x\right\rangle _{t}),
\]
where $R^{\nabla W}$ denotes Taylor's remainder term, and find
\begin{eqnarray}
(\ref{eq:pp}) & = & 2\left|\Lambda\right|\left|\left\langle \nabla W*|\varphi_{t}|^{2}(x)p_{1}^{\varphi_{t}}\cdot\frac{p-\left\langle p\right\rangle _{t}}{\rho}\right\rangle _{t}\right|\nonumber \\
 & = & 2\left|\Lambda\right|\left|\int dy\,\nabla W(\left\langle x\right\rangle _{t}-y)|\varphi_{t}|^{2}(y)\cdot\left\langle p_{1}^{\varphi_{t}}\frac{p-\left\langle p\right\rangle _{t}}{\rho}\right\rangle _{t}\right|\label{eq:pp_1}\\
 &  & +2\left|\Lambda\right|\left|\left\langle \int dy\,\left[R^{\nabla W}(x,y,\left\langle x\right\rangle _{t})\,(x-\left\langle x\right\rangle _{t})\right]|\varphi_{t}|^{2}(y)p_{1}^{\varphi_{t}}\cdot\frac{p-\left\langle p\right\rangle _{t}}{\rho}\right\rangle _{t}\right|.\label{eq:pp_2}
\end{eqnarray}
To estimate term (\ref{eq:pp_1}) we use definition (\ref{eq:alpha_def})
to get an auxiliary bound
\begin{eqnarray*}
\left|\left\langle p_{1}^{\varphi_{t}}\frac{p-\left\langle p\right\rangle _{t}}{\rho}\right\rangle _{t}\right| & = & \left|\left\langle \frac{p-\left\langle p\right\rangle _{t}}{\rho}\right\rangle _{t}-\left\langle q_{1}^{\varphi_{t}}\frac{p-\left\langle p\right\rangle _{t}}{\rho}\right\rangle _{t}\right|=\left|\left\langle q_{1}^{\varphi_{t}}\frac{p-\left\langle p\right\rangle _{t}}{\rho}\right\rangle _{t}\right|\\
 & \leq & \left\Vert q_{1}^{\varphi_{t}}\Psi_{t}\right\Vert _{2}\left\Vert \frac{p-\left\langle p\right\rangle _{t}}{\rho}\Psi_{t}\right\Vert _{2}\\
 & \leq & \frac{1}{\left|\Lambda\right|^{1/2}}\sqrt{\left|\Lambda\right|\left\langle q_{1}^{\varphi}\right\rangle _{t}\left\langle \left(\frac{p-\left\langle p\right\rangle _{t}}{\rho}\right)^{2}\right\rangle _{t}}\\
 & \leq & \frac{\alpha_{t}}{\left|\Lambda\right|^{1/2}},
\end{eqnarray*}
which implies
\begin{eqnarray}
(\ref{eq:pp_1}) & \leq & 2\left|\Lambda\right|\left|\int dy\,\nabla W(\left\langle x\right\rangle _{t}-y)|\varphi_{t}|^{2}(y)\right|\frac{\alpha_{t}}{\left|\Lambda\right|^{1/2}}\nonumber \\
 & \leq & 4\left|\Lambda\right|\left(\left\Vert \nabla W\right\Vert _{\infty}+\left\Vert \nabla W\right\Vert _{1}\right)\left(\left\Vert e^{i\Delta t}\phi^{(0)}\right\Vert _{\infty}^{2}+\left\Vert \varphi_{t}-e^{i\Delta t}\phi^{(0)}\right\Vert _{2}^{2}\right)\frac{\alpha_{t}}{\left|\Lambda\right|^{1/2}}\nonumber \\
 & \leq & C\, C_{\mathrm{prop}}(t)^{2}\frac{\alpha_{t}}{\left|\Lambda\right|^{1/2}}.\label{eq:pp_1_est}
\end{eqnarray}
Here we have inserted the identity
\[
\varphi_{t}=e^{i\Delta t}\phi^{(0)}+\varphi_{t}-e^{i\Delta t}\phi^{(0)}
\]
and used \lemref{proagation}, which implies that
\[
\left\Vert e^{i\Delta t}\phi^{(0)}\right\Vert _{\infty}+\left\Vert \varphi_{t}-e^{i\Delta t}\phi^{(0)}\right\Vert _{2}\leq C_{\mathrm{prop}}(t)\left|\Lambda\right|^{-1/2}.
\]
To control term (\ref{eq:pp_2}) we make use of definition (\ref{eq:alpha_def})
and \lemref{proagation} again and find that
\begin{eqnarray}
(\ref{eq:pp_2}) & \leq & 2\left|\Lambda\right|\left\Vert \int dy\, R^{\nabla W}(x,y,\left\langle x\right\rangle _{t})|\varphi_{t}|^{2}(y)\right\Vert \left\Vert (x-\left\langle x\right\rangle _{t})\Psi_{t}\right\Vert _{2}\left\Vert p_{1}^{\varphi_{t}}\cdot\frac{p-\left\langle p\right\rangle _{t}}{\rho}\Psi_{t}\right\Vert _{2}\nonumber \\
 & \leq & 2\left|\Lambda\right|\left(\sup_{x\in\mathbb{R}^{3}}\left|\int dy\, R^{\nabla W}(x,y,\left\langle x\right\rangle _{t})\left|\varphi_{t}\right|{}^{2}(y)\right|\right)\sqrt{\left\langle (x-\left\langle x\right\rangle _{t})^{2}\right\rangle _{t}\left\langle \left(\frac{p-\left\langle p\right\rangle _{t}}{\rho}\right)^{2}\right\rangle _{t}}\nonumber \\
 & \leq & 4\left|\Lambda\right|\sup_{x,z\in\mathbb{R}^{3}}\left(\left\Vert R^{\nabla W}(x,\cdot,z)\right\Vert _{\infty}+\left\Vert R^{\nabla W}(x,\cdot,z)\right\Vert _{1}\right)\left(\left\Vert e^{i\Delta t}\phi^{(0)}\right\Vert _{\infty}^{2}+\left\Vert \varphi_{t}-e^{i\Delta t}\phi^{(0)}\right\Vert _{2}^{2}\right)\alpha_{t}\nonumber \\
 & \leq & C\, C_{\mathrm{prop}}(t)^{2}\alpha_{t}.\label{eq:pp_2_est}
\end{eqnarray}
Estimates (\ref{eq:pp_1_est}) and (\ref{eq:pp_2_est}) imply the
upper bound
\begin{eqnarray}
(\ref{eq:pp}) & = & (\ref{eq:pp_1})+(\ref{eq:pp_2})\nonumber \\
 & \leq & C\, C_{\mathrm{prop}}(t)^{2}\alpha_{t}.\label{eq:p_var_2}
\end{eqnarray}
  The term (\ref{eq:q_p}) 
is estimated similarly:
\begin{eqnarray}
(\ref{eq:q_p}) & = & 2\left|\Lambda\right|\left|\left\langle q_{1}^{\varphi_{t}}\nabla W(x-y_{1})p_{1}^{\varphi_{t}}\cdot\frac{p-\left\langle p\right\rangle _{t}}{\rho}\right\rangle _{t}\right|\nonumber \\
 & \leq & 2\left|\Lambda\right|\left\Vert q_{1}^{\varphi_{t}}\Psi_{t}\right\Vert _{2}\left\Vert \nabla W(x-y_{1})p_{1}^{\varphi_{t}}\right\Vert \left\Vert \frac{p-\left\langle p\right\rangle _{t}}{\rho}\Psi_{t}\right\Vert _{2}\nonumber \\
 & \leq & 2\left|\Lambda\right|\left|\Lambda\right|^{-1/2}\sqrt{\left|\Lambda\right|\left\langle q_{1}^{\varphi_{t}}\right\rangle _{t}}\left\Vert \nabla W(x-y_{1})p_{1}^{\varphi_{t}}\right\Vert \sqrt{\left\langle \left(\frac{p-\left\langle p\right\rangle _{t}}{\rho}\right)^{2}\right\rangle _{t}}\nonumber \\
 & \leq & C\left|\Lambda\right|\left|\Lambda\right|^{-1/2}\sqrt{\alpha_{t}}\sup_{x\in\mathbb{R}^{3}}\left(\left\Vert \nabla W(x-\cdot)\right\Vert _{\infty}+\left\Vert \nabla W(x-\cdot)\right\Vert _{2}\right)\left(\left\Vert e^{i\Delta t}\phi^{(0)}\right\Vert _{\infty}+\left\Vert \varphi_{t}-e^{i\Delta t}\phi^{(0)}\right\Vert _{2}\right)\sqrt{\alpha_{t}}\nonumber \\
 & \leq & C\left|\Lambda\right|\left|\Lambda\right|^{-1/2}\sqrt{\alpha_{t}}\, C_{\mathrm{prop}}(t)\left|\Lambda\right|^{-1/2}\sqrt{\alpha_{t}}\nonumber \\
 & \leq & C\, C_{\mathrm{prop}}(t)\alpha_{t},\label{eq:p_var_3}
\end{eqnarray}
where we have used definition (\ref{eq:alpha_def}), \lemref{proagation} and the following straightforward
estimate
\begin{eqnarray*}
\left\Vert \nabla W(x-y_{1})p_{1}^{\varphi_{t}}\right\Vert ^{2} & = & \sup_{\left\Vert \chi\right\Vert _{2}=1}\left\langle \chi,p_{1}^{\varphi_{t}}\left[\nabla W(x-y_{1})\right]^{2}p_{1}^{\varphi_{t}}\chi\right\rangle \\
 & = & \sup_{\left\Vert \chi\right\Vert _{2}=1}\left\langle \chi,p_{1}^{\varphi_{t}}\int dy\,\left[\nabla W(x-y)\right]^{2}\left|\varphi_{t}\right|^{2}(y)p_{1}^{\varphi_{t}}\chi\right\rangle \\
 & \leq & \sup_{x\in\mathbb{R}^{3}}\int dy\,\left[\nabla W(x-y)\right]^{2}\left|\varphi_{t}\right|^{2}(y).
\end{eqnarray*}
Our estimate on term (\ref{eq:p+q_q}) 
is more subtle. Due to symmetry in
the gas degrees of freedom, definition (\ref{eq:alpha_def}) of $\alpha_{t}$, and
(\ref{eq:rho_def}) one finds that
\begin{eqnarray}
(\ref{eq:p+q_q}) & = & 2\left|\Lambda\right|\left|\left\langle \nabla W(x-y_{1})q_{1}^{\varphi_{t}}\cdot\frac{p-\left\langle p\right\rangle _{t}}{\rho}\right\rangle _{t}\right|\nonumber \\
 & = & 2\left|\Lambda\right|\left|\left\langle \frac{1}{N}\sum_{k=1}^{N}\nabla W(x-y_{k})q_{k}^{\varphi_{t}}\cdot\frac{p-\left\langle p\right\rangle _{t}}{\rho}\right\rangle _{t}\right|\nonumber \\
 & \leq & 2\rho^{-1}\left\Vert \sum_{k=1}^{N}q_{k}^{\varphi_{t}}\nabla W(x-y_{k})\Psi_{t}\right\Vert _{2}\left\Vert \frac{p-\left\langle p\right\rangle _{t}}{\rho}\Psi_{t}\right\Vert _{2}\nonumber \\
 & \leq & 2\rho^{-1}\left(\left\langle \sum_{k,j=1}^{N}\nabla W(x-y_{k})q_{k}^{\varphi_{t}}q_{j}^{\varphi_{t}}\cdot\nabla W(x-y_{j})\right\rangle _{t}\right)^{1/2}\sqrt{\alpha_{t}}\nonumber \\
 & \leq & \rho^{-2}\left\langle \sum_{k,j=1}^{N}\nabla W(x-y_{k})q_{k}^{\varphi_{t}}q_{j}^{\varphi_{t}}\cdot\nabla W(x-y_{j})\right\rangle _{t}+\alpha_{t}.\label{eq:p_var_3_1}
\end{eqnarray}
In the next step, the sum is split into a sum over diagonal
terms and one over cross terms, i.e.,
\begin{eqnarray}
 &  & \rho^{-2}\left\langle \sum_{k,j=1}^{N}\nabla W(x-y_{k})q_{k}^{\varphi_{t}}q_{j}^{\varphi_{t}}\cdot\nabla W(x-y_{j})\right\rangle _{t}\nonumber \\
 & = & \rho^{-2}N\left\langle \nabla W(x-y_{1})q_{1}^{\varphi_{t}}\cdot\nabla W(x-y_{1})\right\rangle _{t}\label{eq:diagonal}\\
 &  & +\rho^{-2}N(N-1)\left\langle \nabla W(x-y_{1})q_{1}^{\varphi_{t}}q_{2}^{\varphi_{t}}\cdot\nabla W(x-y_{2})\right\rangle _{t}.\label{eq:sym_trick}
\end{eqnarray}
The diagonal term (\ref{eq:diagonal}) can be bounded by
\begin{eqnarray}
(\ref{eq:diagonal}) & \leq & \rho^{-2}N\left\Vert \nabla W(x-y_{1})\Psi_{t}\right\Vert _{2}\left\Vert q_{1}^{\varphi_{t}}\right\Vert \left\Vert \nabla W(x-y_{1})\Psi_{t}\right\Vert _{2}\nonumber \\
 & = & \rho^{-2}N\left\langle \left(\nabla W(x-y_{1})\right)^{2}\right\rangle _{t}\nonumber \\
 & = & \rho^{-2}N\left\langle p_{1}^{\varphi_{t}}\left(\nabla W(x-y_{1})\right)^{2}p_{1}^{\varphi_{t}}\right\rangle _{t}\label{eq:gradW_pp}\\
 &  & +2\rho^{-2}N\,\Im\left\langle p_{1}^{\varphi_{t}}\left(\nabla W(x-y_{1})\right)^{2}q_{1}^{\varphi_{t}}\right\rangle _{t}\label{eq:gradW_pq}\\
 &  & +\rho^{-2}N\left\langle q_{1}^{\varphi_{t}}\left(\nabla W(x-y_{1})\right)^{2}q_{1}^{\varphi_{t}}\right\rangle _{t},\label{eq:gradW_qq}
\end{eqnarray}
where we have again inserted the identity $\mathbbm{1}=p_{1}^{\varphi_{t}}+q_{1}^{\varphi_{t}}$.
Using definition (\ref{eq:alpha_def}) of $\alpha_{t}$, as well as \lemref{proagation},
we compute
\begin{eqnarray}
(\ref{eq:gradW_pp}) & \leq & C\rho^{-2}N\sup_{x\in\mathbb{R}^{3}}\left(\left\Vert \nabla W(x-\cdot)\right\Vert _{\infty}^{2}+\left\Vert \nabla W(x-\cdot)\right\Vert _{2}^{2}\right)\left(\left\Vert e^{i\Delta t}\phi^{(0)}\right\Vert _{\infty}^{2}+\left\Vert \varphi_{t}-e^{i\Delta t}\phi^{(0)}\right\Vert _{2}^{2}\right)\nonumber \\
 & \leq & C\, C_{\mathrm{prop}}(t)^{2}\rho^{-1},\label{eq:p_var_4}
\end{eqnarray}
and
\begin{eqnarray}
(\ref{eq:gradW_pq}) & = & 2\rho^{-2}N\,\Im\left\langle p_{1}^{\varphi_{t}}\left(\nabla W(x-y_{1})\right)^{2}q_{1}^{\varphi_{t}}\right\rangle _{t}\nonumber \\
 & \leq & 2\rho^{-2}N\left\Vert \left(\nabla W(x-y_{1})\right)^{2}p_{1}^{\varphi_{t}}\Psi_{t}\right\Vert _{2}\left\Vert q_{1}^{\varphi_{t}}\Psi_{t}\right\Vert _{2}\nonumber \\
 & \leq & C\rho^{-2}N\,\sup_{x\in\mathbb{R}^{3}}\left(\left\Vert \left(\nabla W(x-\cdot)\right)^{2}\right\Vert _{\infty}+\left\Vert \left(\nabla W(x-\cdot)\right)^{2}\right\Vert _{2}\right)\times\nonumber \\
 &  & \times\left(\left\Vert e^{i\Delta t}\phi^{(0)}\right\Vert _{\infty}+\left\Vert \varphi_{t}-e^{i\Delta t}\phi^{(0)}\right\Vert _{2}\right)\left|\Lambda\right|^{-1/2}\sqrt{\alpha_{t}}\nonumber \\
 & \leq & C\,\rho^{-2}NC_{\mathrm{prop}}(t)\left|\Lambda\right|^{-1/2}\left|\Lambda\right|^{-1/2}\sqrt{\alpha_{t}}\nonumber \\
 & \leq & C\,\rho^{-1}C_{\mathrm{prop}}(t)\sqrt{\alpha_{t}}\nonumber \\
 & \leq & C\, C_{\mathrm{prop}}(t)\left(\alpha_{t}+\rho^{-2}\right),\label{eq:p_var_5}
\end{eqnarray}
and finally
\begin{eqnarray}
(\ref{eq:gradW_qq}) & = & \rho^{-2}N\left\langle q_{1}^{\varphi_{t}}\left(\nabla W(x-y_{1})\right)^{2}q_{1}^{\varphi_{t}}\right\rangle _{t}\nonumber \\
 & \leq & \rho^{-2}N\left|\Lambda\right|^{-1/2}\sqrt{\alpha_{t}}\left\Vert \left(\nabla W\right)^{2}\right\Vert _{\infty}\left|\Lambda\right|^{-1/2}\sqrt{\alpha_{t}}\nonumber \\
 & \leq & C\rho^{-1}\alpha_{t}.\label{eq:p_var_6}
\end{eqnarray}
To estimate the cross terms in (\ref{eq:p_var_3_1}) we again insert the identity $\mathbbm{1}=p_{1}^{\varphi_{t}}+q_{1}^{\varphi_{t}}$ on the
right and the identity $\mathbbm{1}=p_{2}^{\varphi_{t}}+q_{2}^{\varphi_{t}}$ on the left
of the expectation value (\ref{eq:sym_trick}). This yields
\begin{eqnarray}
(\ref{eq:sym_trick}) & = & \rho^{-2}N(N-1)\left|\left\langle (p_{1}^{\varphi_{t}}+q_{1}^{\varphi_{t}})q_{2}^{\varphi_{t}}\nabla W(x-y_{1})\nabla W(x-y_{2})q_{1}^{\varphi_{t}}(p_{2}^{\varphi_{t}}+q_{2}^{\varphi_{t}})\right\rangle _{t}\right|\nonumber \\
 & \leq & \rho^{-2}N(N-1)\left|\left\langle p_{1}^{\varphi_{t}}q_{2}^{\varphi_{t}}\nabla W(x-y_{1})\nabla W(x-y_{2})q_{1}^{\varphi_{t}}p_{2}^{\varphi_{t}}\right\rangle _{t}\right|\label{eq:cross_pqqp}\\
 &  & +2\rho^{-2}N(N-1)\left|\left\langle p_{1}^{\varphi_{t}}q_{2}^{\varphi_{t}}\nabla W(x-y_{1})\nabla W(x-y_{2})q_{1}^{\varphi_{t}}q_{2}^{\varphi_{t}}\right\rangle _{t}\right|\label{eq:cross_pqqq}\\
 &  & +\rho^{-2}N(N-1)\left|\left\langle q_{1}^{\varphi_{t}}q_{2}^{\varphi_{t}}\nabla W(x-y_{1})\nabla W(x-y_{2})q_{1}^{\varphi_{t}}q_{2}^{\varphi_{t}}\right\rangle _{t}\right|.\label{eq:cross_qqqq}
\end{eqnarray}
Each term on the right-hand side can be estimated using definition (\ref{eq:alpha_def})
and \lemref{proagation}:
\begin{eqnarray}
(\ref{eq:cross_pqqp}) & \leq & \rho^{-2}N^{2}\left|\left\langle q_{2}^{\varphi_{t}}p_{1}^{\varphi_{t}}\nabla W(x-y_{1})\nabla W(x-y_{2})p_{2}^{\varphi_{t}}q_{1}^{\varphi_{t}}\right\rangle _{t}\right|\nonumber \\
 & \leq & \rho^{-2}N^{2}\left\Vert q_{2}^{\varphi_{t}}\Psi_{t}\right\Vert _{2}\left\Vert \nabla W(x-y_{1})p_{1}^{\varphi_{t}}\right\Vert \left\Vert \nabla W(x-y_{2})p_{2}^{\varphi_{t}}\right\Vert \left\Vert q_{1}^{\varphi_{t}}\Psi_{t}\right\Vert _{2}\nonumber \\
 & \leq & C\rho^{-2}N^{2}\left|\Lambda\right|^{-1/2}\sqrt{\alpha_{t}}\times\nonumber \\
 &  & \times\left[\sup_{x\in\mathbb{R}^{3}}\left(\left\Vert \nabla W(x-\cdot)\right\Vert _{\infty}+\left\Vert \nabla W(x-\cdot)\right\Vert _{2}\right)\left(\left\Vert e^{i\Delta t}\phi^{(0)}\right\Vert _{\infty}+\left\Vert \varphi_{t}-e^{i\Delta t}\phi^{(0)}\right\Vert _{2}\right)\right]^{2}\times\nonumber \\
 &  & \times\left|\Lambda\right|^{-1/2}\sqrt{\alpha_{t}}\nonumber \\
 & \leq & C\, C_{\mathrm{prop}}(t)^{2}\alpha_{t},\label{eq:p_var_7}
\end{eqnarray}
and
\begin{eqnarray}
(\ref{eq:cross_pqqq}) & \leq & 2\rho^{-2}N^{2}\left|\left\langle q_{2}^{\varphi_{t}}p_{1}^{\varphi_{t}}\nabla W(x-y_{1})\nabla W(x-y_{2})q_{1}^{\varphi_{t}}q_{2}^{\varphi_{t}}\right\rangle _{t}\right|\nonumber \\
 & \leq & 2\rho^{-2}N^{2}\left\Vert q_{2}^{\varphi_{t}}\Psi_{t}\right\Vert _{2}\left\Vert p_{1}^{\varphi_{t}}\nabla W(x-y_{1})\right\Vert \left\Vert \nabla W\right\Vert _{\infty}\left\Vert q_{1}^{\varphi_{t}}q_{2}^{\varphi_{t}}\Psi_{t}\right\Vert _{2}\nonumber \\
 & \leq & C\rho N^{2}\left|\Lambda\right|^{-1/2}\sqrt{\alpha_{t}}\times\nonumber \\
 &  & \times\sup_{x\in\mathbb{R}^{3}}\left(\left\Vert \nabla W(x-\cdot)\right\Vert _{\infty}+\left\Vert \nabla W(x-\cdot)\right\Vert _{2}\right)\left(\left\Vert e^{i\Delta t}\phi^{(0)}\right\Vert _{\infty}+\left\Vert \varphi_{t}-e^{i\Delta t}\phi^{(0)}\right\Vert _{2}\right)\nonumber \\
 &  & \times\left|\Lambda\right|^{-1}\sqrt{\alpha_{t}}\nonumber \\
 & \leq & C\, C_{\mathrm{prop}}(t)\alpha_{t},\label{eq:p_var_8}
\end{eqnarray}
finally
\begin{eqnarray}
(\ref{eq:cross_qqqq}) & \leq & \rho^{-2}N^{2}\left|\left\langle q_{1}^{\varphi_{t}}q_{2}^{\varphi_{t}}\nabla W(x-y_{1})\nabla W(x-y_{2})q_{1}^{\varphi_{t}}q_{2}^{\varphi_{t}}\right\rangle _{t}\right|\nonumber \\
 & \leq & \rho^{-2}N^{2}\left\Vert q_{1}^{\varphi_{t}}q_{2}^{\varphi_{t}}\Psi_{t}\right\Vert _{2}\left\Vert \nabla W\right\Vert _{\infty}^{2}\left\Vert q_{1}^{\varphi_{t}}q_{2}^{\varphi_{t}}\Psi_{t}\right\Vert _{2}\nonumber \\
 & \leq & C\rho^{-2}N^{2}\left|\Lambda\right|^{-1}\sqrt{\alpha_{t}}\left|\Lambda\right|^{-1}\sqrt{\alpha_{t}}\nonumber \\
 & \leq & C\alpha_{t}.\label{eq:p_var_9}
\end{eqnarray}
Combination of estimates (\ref{eq:p_var_1}), (\ref{eq:p_var_2})-(\ref{eq:p_var_3_1}), (\ref{eq:p_var_4})-(\ref{eq:p_var_6}), (\ref{eq:p_var_7})-(\ref{eq:p_var_9})
implies the existence of a $C_{\mathrm{Var}(p)}\in\mathtt{Bounds}$
with
\begin{equation}
(\ref{eq:p_variance})\leq C_{\mathrm{Var}(p)}(t)\left(\alpha_{t}+\rho^{-1}\right).\label{eq:p_var_final}
\end{equation}
With the help of definition (\ref{eq:alpha_def}) of the function $\alpha_{t}$ and \lemref{proagation}
we treat the two remaining terms as follows.\\

\noindent\noun{term }(\ref{eq:1_q})\noun{: }Recall that, by (\ref{eq:intermediate}),
the equation of motion
\begin{equation}
\frac{d}{dt}q_{1}^{\varphi_{t}}=-\frac{d}{dt}p_{1}^{\varphi_{t}}=-i\left[-\Delta_{y_{1}}+W(\left\langle x\right\rangle _{t}-y_{1}),p_{1}^{\varphi_{t}}\right]=i\left[-\Delta_{y_{1}}+W(\left\langle x\right\rangle _{t}-y_{1}),q_{1}^{\varphi_{t}}\right]
\end{equation}
holds. Using the expansion
\[
W(x-y)=W(\left\langle x\right\rangle _{t}-y)+R^{W}(x,y,\left\langle x\right\rangle _{t})\,(x-\left\langle x\right\rangle _{t}),
\]
where $R^{W}(x,y)$ denotes Taylor's remainder, we find
\begin{eqnarray}
(\ref{eq:1_q}) & = & \left|\Lambda\right|\left|\frac{d}{dt}\left\langle q_{1}^{\varphi_{t}}\right\rangle _{t}\right|\nonumber \\
 & = & \left|\Lambda\right|\left|\left\langle \left[W(x-y_{1})-W(\left\langle x\right\rangle _{t}-y_{1}),q_{1}^{\varphi_{t}}\right]\right\rangle _{t}\right|\nonumber \\
 & = & 2\left|\Lambda\right|\left|\Im\left\langle \left(p_{1}^{\varphi_{t}}+q_{1}^{\varphi_{t}}\right)\left(W(x-y_{1})-W(\left\langle x\right\rangle _{t}-y_{1})\right)q_{1}^{\varphi_{t}}\right\rangle _{t}\right|\label{eq:1_q_im_part}\\
 & = & 2\left|\Lambda\right|\left|\Im\left\langle p_{1}^{\varphi_{t}}\left(W(x-y_{1})-W(\left\langle x\right\rangle _{t}-y_{1})\right)q_{1}^{\varphi_{t}}\right\rangle _{t}\right|\label{eq:1_q_im_part2}\\
 & = & 2\left|\Lambda\right|\left|\Im\left\langle p_{1}^{\varphi_{t}}R^{W}(x,y_{1},\left\langle x\right\rangle _{t})\,\left(x-\left\langle x\right\rangle _{t}\right)q_{1}^{\varphi_{t}}\right\rangle _{t}\right|\nonumber \\
 & \leq & 2\left|\Lambda\right|\left\Vert \left(x-\left\langle x\right\rangle _{t}\right)\Psi_{t}\right\Vert _{2}\left\Vert R^{W}(x,y_{1},\left\langle x\right\rangle _{t})p_{1}^{\varphi_{t}}\right\Vert \left\Vert q_{1}^{\varphi_{t}}\Psi_{t}\right\Vert _{2}\nonumber \\
 & \leq & C\left|\Lambda\right|\sqrt{\alpha_{t}}\sup_{x,z\in\mathbb{R}^{3}}\left(\left\Vert R^{W}(x,\cdot,z)\right\Vert _{\infty}+\left\Vert R^{W}(x,\cdot,z)\right\Vert _{2}\right)\left(\left\Vert e^{i\Delta t}\phi^{(0)}\right\Vert _{\infty}+\left\Vert \varphi_{t}-e^{i\Delta t}\phi^{(0)}\right\Vert _{2}\right)\left|\Lambda\right|^{-1/2}\sqrt{\alpha_{t}}\nonumber \\
 & \leq & C\, C_{\mathrm{prop}}(t)\alpha_{t},\label{eq:q1_final}
\end{eqnarray}
where, in the step from (\ref{eq:1_q_im_part}) to (\ref{eq:1_q_im_part2}),
we have used that
\[
q_{1}^{\varphi_{t}}\left(W(x-y_{1})-W(\left\langle x\right\rangle _{t}-y_{1})\right)q_{1}^{\varphi_{t}}
\]
is self-adjoint, so that its expectation value has a vanishing imaginary part.
\\

\noindent\noun{term }(\ref{eq:2_q})\noun{:} 
\begin{eqnarray}
(\ref{eq:2_q}) & = & \left|\Lambda\right|^{2}\left|\frac{d}{dt}\left\langle q_{1}^{\varphi_{t}}q_{2}^{\varphi_{t}}\right\rangle _{t}\right|\nonumber \\
 & = & \left|\Lambda\right|^{2}\left|\left\langle \left[W(x-y_{1})-W(\left\langle x\right\rangle _{t}-y_{1})+W(x-y_{2})-W(\left\langle x\right\rangle _{t}-y_{2}),q_{1}^{\varphi_{t}}q_{2}^{\varphi_{t}}\right]\right\rangle _{t}\right|\nonumber \\
 & = & 4\left|\Lambda\right|^{2}\left|\Im\left\langle q_{2}^{\varphi_{t}}\left(W(x-y_{1})-W(\left\langle x\right\rangle _{t}-y_{1})\right)q_{1}^{\varphi_{t}}q_{2}^{\varphi_{t}}\right\rangle _{t}\right|\label{eq:2_q_1}\\
 & = & 4\left|\Lambda\right|^{2}\left|\Im\left\langle (p_{1}^{\varphi_{t}}+q_{1}^{\varphi_{t}})q_{2}^{\varphi_{t}}\left(W(x-y_{1})-W(\left\langle x\right\rangle _{t}-y_{1})\right)q_{1}^{\varphi_{t}}q_{2}^{\varphi_{t}}\right\rangle _{t}\right|\label{eq:2_q_3}\\
 & = & 4\left|\Lambda\right|^{2}\left|\Im\left\langle p_{1}^{\varphi_{t}}q_{2}^{\varphi_{t}}\left(W(x-y_{1})-W(\left\langle x\right\rangle _{t}-y_{1})\right)q_{1}^{\varphi_{t}}q_{2}^{\varphi_{t}}\right\rangle _{t}\right|\label{eq:2_q_2}\\
 & \leq & 4\left|\Lambda\right|^{2}\left(\left|\left\langle p_{1}^{\varphi_{t}}q_{2}^{\varphi_{t}}W(x-y_{1})q_{1}^{\varphi_{t}}q_{2}^{\varphi_{t}}\right\rangle \right|+\left|\left\langle p_{1}^{\varphi_{t}}q_{2}^{\varphi_{t}}W(\left\langle x\right\rangle _{t}-y_{1})q_{1}^{\varphi_{t}}q_{2}^{\varphi_{t}}\right\rangle \right|\right)\nonumber \\
 & \leq & 4\left|\Lambda\right|^{2}\left\Vert q_{2}^{\varphi_{t}}\Psi_{t}\right\Vert _{2}\left(\left\Vert W(x-y_{1})p_{1}^{\varphi_{t}}\right\Vert +\left\Vert W(\left\langle x\right\rangle _{t}-y_{1})p_{1}^{\varphi_{t}}\right\Vert \right)\left\Vert q_{1}^{\varphi_{t}}q_{2}^{\varphi_{t}}\Psi_{t}\right\Vert _{2}\nonumber \\
 & \leq & C\left|\Lambda\right|\left|\Lambda\right|^{-1/2}\sqrt{\alpha_{t}}\sup_{x\in\mathbb{R}^{3}}\left(\left\Vert W(x-\cdot)\right\Vert _{\infty}+\left\Vert W(x-\cdot)\right\Vert _{2}\right)\times\nonumber \\
 &  & \times\left(\left\Vert e^{i\Delta t}\phi^{(0)}\right\Vert _{\infty}+\left\Vert \varphi_{t}-e^{i\Delta t}\phi^{(0)}\right\Vert _{2}\right)\left|\Lambda\right|^{-1}\sqrt{\alpha_{t}}\nonumber \\
 & \leq & C\, C_{\mathrm{prop}}(t)\alpha_{t}.\label{eq:q2_final}
\end{eqnarray}
Note that (\ref{eq:2_q_1}) holds, because every other operator
inside the expectation value commutes with $q_{2}^{\varphi_{t}}$.
In the step from (\ref{eq:2_q_3}) to (\ref{eq:2_q_2}) we have used
that the operator
\[
q_{1}^{\varphi_{t}}q_{2}^{\varphi_{t}}\left(W(x-y_{1})-W(\left\langle x\right\rangle _{t}-y_{1})\right)q_{1}^{\varphi_{t}}q_{2}^{\varphi_{t}}
\]
is self-adjoint, and hence the imaginary part of its expectation value
vanishes.\\

Upon collecting estimates (\ref{eq:q_var_final}), (\ref{eq:p_var_final}),
(\ref{eq:q1_final}), (\ref{eq:q2_final}), we find $C_{\alpha}^{(1)},C_{\alpha}^{(2)}\in\mathtt{Bounds}$
such that
\[
\frac{d}{dt}\alpha_{t}\leq C_{\alpha}^{(1)}(t)\alpha_{t}+C_{\alpha}^{(2)}(t)\rho^{-1},
\]
which concludes the proof.
\end{proof}

Finally we prove the bound on the time derivative of $\beta_{t}$. The aim is to show that either the corresponding terms can be estimated in terms of the quantities $\alpha_{t}$ and/or $\beta_{t}$, or else
that these terms are small if one of the parameters $|\Lambda|$ or $\rho$ is large. 

In addition to the arguments we have used in our proof of Lemma \ref{lem:alpha}, we exploit the following facts:
\begin{itemize}
  \item The crucial cancellation
\[p_{1}^{\varphi_{t}}\nabla W(x-y_{1})p_{1}^{\varphi_{t}}-\nabla W*|\varphi_{t}|^{2}(x)p_{1}^{\varphi_{t}}=0.\] 
which will take place in (\ref{eq:convolution}) below. This cancellation determines the structure of the macroscopic equations of motion.
  \item 
  The uniform spreading of the wave function $\phi_{t}^{\mathrm{ref}}$, which is controlled by the propagation estimate (ii) in Lemma \ref{lem:proagation}. With its help one can show that the convolution
  \[\nabla W*|\phi_{t}^{(\mathrm{ref})}|^{2}(x)\]
  is small for large $|\Lambda|$.
  \item The propagation estimates on the wave function $\epsilon_{t}$ of the gas excitations w.r.t. $\phi^{(\mathrm{ref})}_{t}$, which are provided by Lemma \ref{lem:ref_and_excitation}.
\end{itemize}

\begin{proof}[\textbf{Proof of \lemref{beta}}]
Recall that $t\mapsto\beta_{t}$ is smooth, so that
\begin{eqnarray}
\frac{d}{dt}\beta_{t} & \leq & \left|\frac{d}{dt}\left(X_{t}-\left\langle x\right\rangle _{t}\right)\right|\label{eq:beta_1}\\
 &  & +\left|\frac{d^{2}}{dt^{2}}\left(X_{t}-\left\langle x\right\rangle _{t}\right)\right|\label{eq:beta_2}\\
 &  & +\left|\frac{d}{dt}\left\langle \left(\phi_{t}^{(\mathrm{ref})}+\epsilon_{t}\right)-\Lambda^{1/2}\varphi_{t},\left(\phi_{t}^{(\mathrm{ref})}+\epsilon_{t}\right)-\Lambda^{1/2}\varphi_{t}\right\rangle \right|^{1/2}.\label{eq:beta_3}
\end{eqnarray}
We shall estimate the terms on the right-hand side individually.\\
\\
\noindent\noun{Term (\ref{eq:beta_1})}: By definition (\ref{eq:beta_def}) of the function $\beta_{t}$
we immediately get
\[
(\ref{eq:beta_1})=\left|\frac{d}{dt}\left(X_{t}-\left\langle x\right\rangle _{t}\right)\right|\leq\beta_{t}.
\]
\noindent\noun{Term (\ref{eq:beta_2})}: Using the equations of motion
(\ref{eq:macro_tracer}) and (\ref{eq:ehrenfest}), we find
\begin{eqnarray*}
(\ref{eq:beta_2}) & = & \left|\frac{d^{2}}{dt^{2}}\left(X_{t}-\left\langle x\right\rangle _{t}\right)\right|\\
 & \leq & 2\left|\left\langle \nabla V(X_{t})-\nabla V(x)+\nabla W*|\epsilon_{t}|^{2}(X_{t})+2\Re\nabla W*\left(\overline{\phi_{t}^{(\mathrm{ref})}}\epsilon_{t}\right)(X_{t})-\frac{1}{\rho}\nabla\sum_{k=1}^{N}W(x-y_{k})\right\rangle _{t}\right|.
\end{eqnarray*}
Using the expansion
\[
\nabla V(x)=\nabla V(X_{t})+R^{\nabla V}(x,X_{t})\,(x-X_{t}),
\]
where $R^{\nabla V}$ denotes Taylor's remainder, we obtain
\begin{eqnarray}
\left|\left\langle \nabla V(X_{t})-\nabla V(x)\right\rangle _{t}\right| & \leq & \left|\left\langle R^{\nabla V}(x,X_{t})\cdot\left(x-\left\langle x\right\rangle _{t}+\left\langle x\right\rangle _{t}-X_{t}\right)\right\rangle _{t}\right|\nonumber \\
 & \leq & \left\Vert R^{\nabla V}\right\Vert _{\infty}\left\Vert \left(x-\left\langle x\right\rangle _{t}\right)\Psi_{t}\right\Vert _{2}+\left\Vert R^{\nabla V}\right\Vert _{\infty}\left|\left\langle x\right\rangle _{t}-X_{t}\right|\nonumber \\
 & \leq & C\sqrt{\alpha_{t}}+C\beta_{t},\label{eq:beta_2_1}
\end{eqnarray}
where we have used (\ref{eq:alpha_def}) and (\ref{eq:beta_def}). Furthermore,
\begin{align}
 & \left|\left\langle \nabla W*|\epsilon_{t}|^{2}(X_{t})+2\Re\nabla W*\left(\overline{\phi_{t}^{(\mathrm{ref})}}\epsilon_{t}\right)(X_{t})-\frac{1}{\rho}\nabla\sum_{k=1}^{N}W(x-y_{k})\right\rangle _{t}\right|\label{eq:gradW_1}\\
= & \left|\left\langle \nabla W*|\phi_{t}^{(\mathrm{ref})}+\epsilon_{t}|^{2}(X_{t})-\nabla W*|\phi_{t}^{(\mathrm{ref})}|^{2}(X_{t})-\left|\Lambda\right|\nabla W(x-y_{1})\right\rangle _{t}\right|\label{eq:gradW_2}\\
\leq & \left|\left\langle \nabla W*|\phi_{t}^{(\mathrm{ref})}+\epsilon_{t}|^{2}(X_{t})-\left|\Lambda\right|\nabla W*|\varphi_{t}|^{2}(X_{t})\right\rangle _{t}\right|\label{eq:phi-varphi}\\
 & +\left|\Lambda\right|\left|\left\langle \nabla W*|\varphi_{t}|^{2}(X_{t})-\nabla W*|\varphi_{t}|^{2}(x)\right\rangle _{t}\right|\label{eq:X-mean_x}\\
 & +\left|\Lambda\right|\left|\left\langle \nabla W*|\varphi_{t}|^{2}(x)-\nabla W(x-y_{1})\right\rangle _{t}\right|\label{eq:convolution}\\
 & +\left|\nabla W*|\phi_{t}^{(\mathrm{ref})}|^{2}(X_{t})\right|.\label{eq:gradient}
\end{align}
All these terms are estimated below. In the step
from (\ref{eq:gradW_1}) to (\ref{eq:gradW_2}), we have used the identity
\[
|\phi_{t}^{(\mathrm{ref})}+\epsilon_{t}|^{2}-|\phi_{t}^{(\mathrm{ref})}|^{2}=|\epsilon_{t}|^{2}+2\Re\overline{\phi_{t}^{(\mathrm{ref})}}\epsilon_{t},
\]
and symmetry in the gas degrees of freedom to replace
\[
\frac{1}{\rho}\nabla\sum_{k=1}^{N}W(x-y_{k})
\]
by
\[
\left|\Lambda\right|\nabla W(x-y_{1}).
\]
To estimate the term (\ref{eq:phi-varphi}), we use the identity
\[
|\phi_{t}^{(\mathrm{ref})}+\epsilon_{t}|^{2}-\left|\Lambda\right||\varphi_{t}|^{2}=\Re\left[\left(\phi_{t}^{(\mathrm{ref})}+\epsilon_{t}+\left|\Lambda\right|^{1/2}\varphi_{t}\right)\overline{\left(\phi_{t}^{(\mathrm{ref})}+\epsilon_{t}-\left|\Lambda\right|^{1/2}\varphi_{t}\right)}\right],
\]
which, together with definition (\ref{eq:beta_def}) and \lemref{proagation},
implies that
\begin{eqnarray}
(\ref{eq:phi-varphi}) & = & \left|\left\langle \int dz\,\nabla W(X_{t}-z)\Re\left[\left(\phi_{t}^{(\mathrm{ref})}(z)+\epsilon_{t}(z)+\left|\Lambda\right|^{1/2}\varphi_{t}(z)\right)\overline{\left(\phi_{t}^{(\mathrm{ref})}(z)+\epsilon_{t}(z)-\left|\Lambda\right|^{1/2}\varphi_{t}(z)\right)}\right]\right\rangle _{t}\right|\nonumber \\
 & \leq & \left\Vert \nabla W(X_{t}-\cdot)\Re\left(\phi_{t}^{(\mathrm{ref})}(\cdot)+\epsilon_{t}(\cdot)-\left|\Lambda\right|^{1/2}\varphi_{t}(\cdot)+2\left|\Lambda\right|^{1/2}\varphi_{t}(\cdot)\right)\right\Vert _{2}\times\label{eq:phi-varphi-1}\\
 & &\qquad \times\left\Vert \phi_{t}^{(\mathrm{ref})}+\epsilon_{t}-\left|\Lambda\right|^{1/2}\varphi_{t}\right\Vert _{2}\nonumber\\
 & \leq & C\left(\left\Vert \nabla W\right\Vert _{\infty}+\left\Vert \nabla W\right\Vert _{2}\right)\times\nonumber \\
 &  & \qquad \times\left(\left\Vert \phi_{t}^{(\mathrm{ref})}+\epsilon_{t}-\left|\Lambda\right|^{1/2}\varphi_{t}\right\Vert _{2}+2\left|\Lambda\right|^{1/2}\left\Vert e^{i\Delta t}\phi^{(0)}\right\Vert _{\infty}+2\left|\Lambda\right|^{1/2}\left\Vert \varphi_{t}-e^{i\Delta t}\phi^{(0)}\right\Vert _{2}\right)\times\label{eq:phi-varphi-2}\\
 &  & \qquad\qquad \times\left\Vert \phi_{t}^{(\mathrm{ref})}+\epsilon_{t}-\left|\Lambda\right|^{1/2}\varphi_{t}\right\Vert _{2}\nonumber \\
 & \leq & C\left(\beta_{t}^{2}+C_{\mathrm{prop}}(t)\beta_{t}\right),\label{eq:beta_2_2}
\end{eqnarray}
where, in the step from (\ref{eq:phi-varphi-1}) to (\ref{eq:phi-varphi-2}),
we have used the identity
\[
\varphi_{t}=e^{i\Delta t}\phi^{(0)}+\varphi_{t}-e^{i\Delta t}\phi^{(0)}.
\]
Next, by expanding $\nabla W$ according to
\[
\nabla W(x-y)=\nabla W(X_{t}-y)+R^{\nabla W}(x,y,X_{t})\,(X_{t}-x),
\]
one gets
\begin{eqnarray}
(\ref{eq:X-mean_x}) & = & \left|\Lambda\right|\left|\left\langle \int dy\,\left[\nabla W(X_{t}-y)-\nabla W(x-y)\right]|\varphi_{t}|^{2}(y)\right\rangle _{t}\right|\nonumber \\
 & = & \left|\Lambda\right|\left|\left\langle \int dy\, R^{\nabla W}(x,y,X_{t})|\varphi_{t}|^{2}(y)\,(X_{t}-\left\langle x\right\rangle _{t}+\left\langle x\right\rangle _{t}-x)\right\rangle _{t}\right|\nonumber \\
 & \leq & \left|\Lambda\right|\left|\left\langle \int dy\, R^{\nabla W}(x,y,X_{t})|\varphi_{t}|^{2}(y)\,(X_{t}-\left\langle x\right\rangle _{t})\right\rangle _{t}\right|+\left|\Lambda\right|\left|\left\langle \int dy\, R^{\nabla W}(x,y,X_{t})|\varphi_{t}|^{2}(y)\,(\left\langle x\right\rangle _{t}-x)\right\rangle _{t}\right|\nonumber \\
 & \leq & \left|\Lambda\right|\sup_{x\in\mathbb{R}^{3}}\left\Vert \int dy\, R^{\nabla W}(x,y,X_{t})|\varphi_{t}|^{2}(y)\right\Vert \left[\left|X_{t}-\left\langle x\right\rangle _{t}\right|+\left\Vert (\left\langle x\right\rangle _{t}-x)\Psi_{t}\right\Vert _{2}\right]\nonumber \\
 & \leq & C\left|\Lambda\right|\sup_{x,z\in\mathbb{R}^{3}}\left(\left\Vert R^{\nabla W}(x,\cdot,z)\right\Vert _{\infty}+\left\Vert R^{\nabla W}(x,\cdot,z)\right\Vert _{1}\right)\left(\left\Vert e^{i\Delta t}\phi^{(0)}\right\Vert _{\infty}^{2}+\left\Vert \varphi_{t}-e^{i\Delta t}\phi^{(0)}\right\Vert _{2}^{2}\right)\left(\beta_{t}+\sqrt{\alpha_{t}}\right)\nonumber \\
 & \leq & C\left|\Lambda\right|\, C_{\mathrm{prop}}(t)^{2}\left|\Lambda\right|^{-1}\left(\beta_{t}+\sqrt{\alpha_{t}}\right)\nonumber \\
 & \leq & C\, C_{\mathrm{prop}}(t)^{2}\left(\beta_{t}+\sqrt{\alpha_{t}}\right),\label{eq:beta_2_3}
\end{eqnarray}
where we have used definitions (\ref{eq:alpha_def}) and (\ref{eq:beta_def}), 
as well as \lemref{proagation}.\\

Moreover, by inserting the identity $\mathbbm{1}=p_{1}^{\varphi_{t}}+q_{1}^{\varphi_{t}}$
and noting that
\[
p_{1}^{\varphi_{t}}\nabla W(x-y_{1})p_{1}^{\varphi_{t}}=\nabla W*|\varphi_{t}|^{2}(x)p_{1}^{\varphi_{t}},
\]
we find that
\begin{eqnarray}
(\ref{eq:convolution}) & = & \left|\Lambda\right|\left|\left\langle (p_{1}^{\varphi_{t}}+q_{1}^{\varphi_{t}})\left[\nabla W*|\varphi_{t}|^{2}(x)-\nabla W(x-y_{1})\right](p_{1}^{\varphi_{t}}+q_{1}^{\varphi_{t}})\right\rangle _{t}\right|\nonumber \\
 & \leq & 2\left|\Lambda\right|\left|\left\langle q_{1}^{\varphi_{t}}\left[\nabla W*|\varphi_{t}|^{2}(x)-\nabla W(x-y_{1})\right]p_{1}^{\varphi_{t}}\right\rangle _{t}\right|\nonumber \\
 &  & +\left|\Lambda\right|\left|\left\langle q_{1}^{\varphi_{t}}\left[\nabla W*|\varphi_{t}|^{2}(x)-\nabla W(x-y_{1})\right]q_{1}^{\varphi_{t}}\right\rangle _{t}\right|\nonumber \\
 & \leq & 2\left|\Lambda\right|\left\Vert q_{1}^{\varphi_{t}}\Psi_{t}\right\Vert _{2}\left(\left\Vert \nabla W*|\varphi_{t}|^{2}(x)\right\Vert +\left\Vert \nabla W(x-y_{1})p_{1}^{\varphi_{t}}\right\Vert \right)\nonumber \\
 &  & +2\left|\Lambda\right|\left\Vert q_{1}^{\varphi_{t}}\Psi_{t}\right\Vert _{2}\left\Vert \nabla W\right\Vert _{\infty}\left\Vert q_{1}^{\varphi_{t}}\Psi_{t}\right\Vert _{2}\nonumber \\
 & \leq & C\left|\Lambda\right|\left|\Lambda\right|^{-1/2}\sqrt{\alpha_{t}}\left(\left\Vert \nabla W\right\Vert _{\infty}+\left\Vert \nabla W\right\Vert _{1}\right)\left(\left\Vert e^{i\Delta t}\phi^{(0)}\right\Vert _{\infty}+\left\Vert \varphi_{t}-e^{i\Delta t}\phi^{(0)}\right\Vert _{2}\right)\nonumber \\
 &  & +2\left|\Lambda\right|\left|\Lambda\right|^{-1/2}\sqrt{\alpha_{t}}\left\Vert \nabla W\right\Vert _{\infty}\left|\Lambda\right|^{-1/2}\sqrt{\alpha_{t}}\nonumber \\
 & \leq & C\, C_{\mathrm{prop}}(t)\sqrt{\alpha_{t}}+C\alpha_{t}.\label{eq:beta_2_4}
\end{eqnarray}

Finally, by integrating by parts, one arrives at the estimate
\begin{eqnarray*}
(\ref{eq:gradient}) & = & \left|\nabla W*|\phi_{t}^{(\mathrm{ref})}|^{2}(X_{t})\right|\\
 & \leq & 2\left\Vert W\right\Vert _{1}\left\Vert \phi_{t}^{(\mathrm{ref})}\right\Vert _{\infty}\left\Vert \nabla\phi_{t}^{(\mathrm{ref})}\right\Vert _{\infty},
\end{eqnarray*}
which, with the help of \lemref{proagation}, yields
\begin{equation}
(\ref{eq:gradient})\leq2C_{\mathrm{ref}}^{2}\left\Vert W\right\Vert _{1}\left|\Lambda\right|^{-1/3}.\label{eq:beta_2_5}
\end{equation}

The estimates (\ref{eq:beta_2_1}), (\ref{eq:beta_2_2}), (\ref{eq:beta_2_3}),
(\ref{eq:beta_2_4}), and (\ref{eq:beta_2_5}) imply the existence
of a function $C_{\mathrm{vel}}\in\mathtt{Bounds}$ such that
\begin{equation}
(\ref{eq:beta_2})\leq C_{\mathrm{vel}}(t)\left(\beta_{t}+\beta_{t}^{2}+\sqrt{\alpha_{t}}+\alpha_{t}+\left|\Lambda\right|^{-1/3}\right).\label{eq:beta_2_final}
\end{equation}
\noindent\noun{Term (\ref{eq:beta_3})}: With the help of the equations
of motion (\ref{eq:macro_gas}) and (\ref{eq:intermediate}) we find that
\begin{eqnarray}
(\ref{eq:beta_3}) & = & \left|\frac{d}{dt}\left\langle \left(\phi_{t}^{(\mathrm{ref})}+\epsilon_{t}\right)-\Lambda^{1/2}\varphi_{t},\left(\phi_{t}^{(\mathrm{ref})}+\epsilon_{t}\right)-\Lambda^{1/2}\varphi_{t}\right\rangle \right|^{1/2}\nonumber \\
 & \leq & \left|2\Im\left\langle \left(-\Delta_{y}+W(X_{t}-\cdot)\right)\left(\phi_{t}^{(\mathrm{ref})}+\epsilon_{t}\right)-\left(-\Delta_{y}+W(\left\langle x\right\rangle _{t}-\cdot)\right)\Lambda^{1/2}\varphi_{t},\left(\phi_{t}^{(\mathrm{ref})}+\epsilon_{t}\right)-\Lambda^{1/2}\varphi_{t}\right\rangle \right|^{1/2}\nonumber \\
 & = & \bigg|2\Im\left\langle \left(-\Delta_{y}+W(X_{t}-\cdot)\right)\left(\phi_{t}^{(\mathrm{ref})}+\epsilon_{t}-\Lambda^{1/2}\varphi_{t}\right),\phi_{t}^{(\mathrm{ref})}+\epsilon_{t}-\Lambda^{1/2}\varphi_{t}\right\rangle \label{eq:zees_is_sero}\\
 &  & +2\Im\left\langle \left(W(X_{t}-\cdot)-W(\left\langle x\right\rangle _{t}-\cdot)\right)\Lambda^{1/2}\varphi_{t},\left(\phi_{t}^{(\mathrm{ref})}+\epsilon_{t}\right)-\Lambda^{1/2}\varphi_{t}\right\rangle \bigg|^{1/2}.\nonumber 
\end{eqnarray}
Because the operator 
\[
-\Delta_{y}+W(X_{t}-y)
\]
is self-adjoint, the term in (\ref{eq:zees_is_sero}) is zero, so that
\begin{eqnarray}
(\ref{eq:beta_3}) & \leq & \left|2\Im\left\langle \left(W(X_{t}-\cdot)-W(\left\langle x\right\rangle _{t}-\cdot)\right)\Lambda^{1/2}\varphi_{t},\phi_{t}^{(\mathrm{ref})}+\epsilon_{t}-\Lambda^{1/2}\varphi_{t}\right\rangle \right|^{1/2}\nonumber \\
 & \leq & 2\left\Vert \left(W(X_{t}-\cdot)-W(\left\langle x\right\rangle _{t}-\cdot)\right)\Lambda^{1/2}\varphi_{t}\right\Vert _{2}^{1/2}\left\Vert \phi_{t}^{(\mathrm{ref})}+\epsilon_{t}-\Lambda^{1/2}\varphi_{t}\right\Vert _{2}^{1/2}\nonumber \\
 & \leq & 2\left|\Lambda\right|^{1/2}\left\Vert \left(W(X_{t}-\cdot)-W(\left\langle x\right\rangle _{t}-\cdot)\right)\varphi_{t}\right\Vert _{2}+2\beta_{t}.\label{eq:beta_3_1}
\end{eqnarray}
We expand $W$ according to
\[
W(\left\langle x\right\rangle _{t}-y)=W(X_{t}-y)+R^{W}(\left\langle x\right\rangle _{t},y,X_{t})\,(\left\langle x\right\rangle _{t}-X_{t}),
\]
and use estimate (\ref{eq:beta_3})
to obtain
\begin{eqnarray}
(\ref{eq:beta_3_1}) & = & 2\left|\Lambda\right|^{1/2}\left\Vert R^{W}(\left\langle x\right\rangle _{t},\cdot,X_{t})\varphi_{t}\,(X_{t}-\left\langle x\right\rangle _{t})\right\Vert _{2}+2\beta_{t}\nonumber \\
 & \leq & C\left|\Lambda\right|^{1/2}\beta_{t}\sup_{x,z\in\mathbb{R}^{3}}\left(\left\Vert R^{W}(x,\cdot,z)\right\Vert _{\infty}+\left\Vert R^{W}(x,\cdot,z)\right\Vert _{2}\right)\left(\left\Vert e^{i\Delta t}\phi^{(0)}\right\Vert _{\infty}+\left\Vert \varphi_{t}-e^{i\Delta t}\phi^{(0)}\right\Vert _{2}\right)+2\beta_{t}\nonumber \\
 & \leq & C(1+C_{\mathrm{prop}}(t))\beta_{t},\label{eq:beta_3_final}
\end{eqnarray}
where, once again, we have used (\ref{eq:beta_def}) and \lemref{proagation}.\\

Given estimates (\ref{eq:beta_2_final}) and (\ref{eq:beta_3_final}),
one infers that
\[
\frac{d}{dt}\beta_{t}\leq C_{\beta}^{(1)}(t)\left(\beta_{t}+\beta_{t}^{2}\right)+C_{\beta}^{(2)}(t)\left(\sqrt{\alpha_{t}}+\alpha_{t}+\left|\Lambda\right|^{-1/3}\right),
\]
for $C_{\beta}^{(1)},C_{\beta}^{(2)}\in\mathtt{Bounds}$. This concludes the proof of Lemma \ref{lem:beta}.\end{proof}
\begin{rem}\label{rem:many-tracer}
Note that our analysis easily generalizes to systems of $M > 1$
tracer particles with a microscopic Hamiltonian of, for example, the following form
\begin{equation}
H:=-\sum_{k=1}^{M}\left(\frac{\Delta_{x_{k}}}{2\rho}+\rho V(x_{k})+\sum_{j=1}^{N}W(x_{k}-y_{j})\right)+\sum_{1\leq j<k\leq M}\rho\, I(x_{k}-x_{j})- \sum_{k=1}^{N} \Delta_{y_k}, \label{eq:many_tracer_micro_H}
\end{equation}
where we denote by $x_{1},x_{2},\ldots,x_{M}$ and $p_{1},p_{2},\ldots,p_{M}$
the positions and momenta of the $M$ tracer particles,
and  by $I$ a regular pair potential. To prove a result analogous
 to \thmref{main}, comparing the microscopic dynamics generated
by (\ref{eq:many_tracer_micro_H}) to the macroscopic dynamics,
\begin{eqnarray}
i\frac{d}{dt}\epsilon_{t}(y) & = & \left(-\Delta_{y}+\sum_{k=1}^{M}W\left(X_{k,t}-y\right)\right)\epsilon_{t}(y)+\sum_{k=1}^{M}W(X_{k,t}-y)\phi_{t}^{(\mathrm{ref})}(y),\label{eq:macro_gas-1}\\
\frac{d^{2}X_{k,t}}{dt^{2}} & = & -\nabla V(X_{k,t})-\sum_{j\neq k}\nabla I(X_{k,t}-X_{j,t})-\nabla W*|\epsilon_{t}|^{2}(X_{k,t})-2\Re\nabla W*\left(\overline{\phi_{t}^{(\mathrm{ref})}}\epsilon_{t}\right)(X_{k,t}),\label{eq:macro_tracer-1}
\end{eqnarray}
for all $1\leq k\leq M$, the following natural adaptations are
needed:
\begin{enumerate}
\item As in \defref{initial_values}, we need to assume the initial wave
functions 
\[
\Psi^{(0)}(x_{1},x_{2},\ldots,x_{M},y_{1},y_{2},\ldots,y_{N})
\]
to be given by a product of a localized wave packet for the tracer particles
-- compare to $\chi^{(0)}$ in (\ref{eq:initial_spread}) -- and a product wave function $\prod_{k=1}^{N}\phi^{(0)}(y_{k})$ for the gas particles.
\item The appropriate intermediate dynamics of the effective wave function
of a gas particle (\ref{eq:intermediate}) is of the form 
\[
i\partial_{t}\varphi_{t}(y)=\left(-\Delta_{y}+\sum_{k=1}^{M}W\left(\left\langle x_{k}\right\rangle _{t}-y\right)\right)\varphi_{t}(y).
\]

\item The quantity $\alpha_{t}$ defined in (\ref{eq:alpha_def}) must be
replaced by
\[
\alpha_{t}:=\sqrt{\sum_{k=1}^{M}\left\langle \left(x_{k}-\left\langle x_{k}\right\rangle _{t}\right)^{2}\right\rangle _{t}^{2}+\sum_{k=1}^{M}\left\langle \left(\frac{p_{k}-\left\langle p_{k}\right\rangle _{t}}{\rho}\right)^{2}\right\rangle _{t}^{2}+\left(\left|\Lambda\right|\left\langle q_{1}^{\varphi_{t}}\right\rangle _{t}\right)^{2}+\left(\left|\Lambda\right|^{2}\left\langle q_{1}^{\varphi_{t}}q_{2}^{\varphi_{t}}\right\rangle _{t}\right)^{2}}.
\]

\item The quantity $\beta_{t}$ defined in (\ref{eq:beta_def}) must be
replaced by
\[
\beta_{t}:=\sqrt{\sum_{k=1}^{M}\left(X_{k,t}-\left\langle x_{k}\right\rangle _{t}\right)^{2}+\sum_{k=1}^{M}\left(\frac{d\left(X_{k,t}-\left\langle x_{k}\right\rangle _{t}\right)}{dt}\right)^{2}+\left\Vert \left(\phi_{t}^{(\mathrm{ref})}+\epsilon_{t}\right)-\left|\Lambda\right|^{1/2}\varphi_{t}\right\Vert _{2}^{2}}.
\]

\end{enumerate}
Taking advantage of the commutation relations 
\[
\lbrack x_k,p_j \rbrack = 0,\qquad \forall k \not= j,
\]
\lemref{alpha} and \lemref{beta} can be proven along the same lines
as demonstrated for a single tracer particle. The only new terms
are ones depending on the pair potential $I$, and they are of
the form
\begin{equation}
\left|\left\langle \left\{ \nabla I(x_{k}-x_{j}),\frac{p_{k}-\left\langle p_{k}\right\rangle _{t}}{\rho}\right\} \right\rangle _{t}\right|.\label{eq:int_term}
\end{equation}
Using an expansion of $\nabla I$ of the form 
\[
\nabla I(x_{k}-x_{j})=\nabla I\left(\left\langle x_{k}\right\rangle _{t}-\left\langle x_{j}\right\rangle _{t}\right)+R_{k}^{\nabla I}(x_{k},x_{j},\left\langle x_{k}\right\rangle _{t})\,\left(x_{k}-\left\langle x_{k}\right\rangle _{t}\right)+R_{j}^{\nabla I}(\left\langle x_{k}\right\rangle _{t},x_{j},\left\langle x_{j}\right\rangle _{t})\,\left(x_{j}-\left\langle x_{j}\right\rangle _{t}\right),
\]
where $R_{k}^{\nabla I}$ denote Taylor's remainder terms, one can treat (\ref{eq:int_term}) in the same way as the term (\ref{eq:grad_V})
in the case of only one tracer particle.
\end{rem}
\appendix

\section{Propagation Estimates}
\begin{lem}
\label{lem:proagation}There are $C_{\mathrm{prop}},C_{\mathrm{ref}}\in\mathtt{Bounds}$
such that:\end{lem}
\begin{enumerate}[label=(\roman*)]
\item The solution $t\mapsto\varphi_{t}$ to (\ref{eq:intermediate}) with
initial value (\ref{eq:intermediate initial value}) fulfills
\[
\left\Vert e^{i\Delta t}\phi^{(0)}\right\Vert _{\infty}+\left\Vert \varphi_{t}-e^{i\Delta t}\phi^{(0)}\right\Vert _{2}\leq C_{\mathrm{prop}}(t)\left|\Lambda\right|^{-1/2}
\]
for all times $t\geq0$.
\item The solution $t\mapsto\phi_{t}^{(\mathrm{ref})}$ to (\ref{eq:free_reference})
with initial value (\ref{eq:initial_reference}) fulfills
\begin{equation}
\left\Vert \phi_{t}^{(\mathrm{ref})}\right\Vert _{\infty}\leq C_{\mathrm{ref}},\qquad\left\Vert \nabla\phi_{t}^{(\mathrm{ref})}\right\Vert _{\infty}\leq C_{\mathrm{ref}}\left|\Lambda\right|^{-1/3}.\label{eq:phi_ref_prop_est}
\end{equation}
\end{enumerate}
\begin{proof}
~
\begin{enumerate}[label=(\roman*)]
\item Because of (\ref{eq:intermediate initial value}) and (\ref{eq:initial_gas_wf})
one immediately gets the estimate
\begin{equation}
\left\Vert e^{i\Delta t}\phi^{(0)}\right\Vert _{\infty}\leq\left\Vert \widehat{e^{i\Delta t}\phi^{(0)}}\right\Vert _{1}=\left\Vert \widehat{\phi^{(0)}}\right\Vert _{1}=C\left|\Lambda\right|^{-1/2}\label{eq:free estimate}
\end{equation}
for all $t\in\mathbb{R}$. Moreover, any solution to (\ref{eq:intermediate})
fulfills the integral equation
\[
\varphi_{t}=e^{i\Delta t}\phi^{(0)}-i\int_{0}^{t}ds\, e^{i\Delta(t-s)}\left[W\left(\left\langle x\right\rangle _{s}-\cdot\right)\right]\varphi_{s}.
\]
Hence, we infer the estimates
\begin{eqnarray*}
\left\Vert \varphi_{t}-e^{i\Delta t}\phi^{(0)}\right\Vert _{2} & = & \left\Vert \int_{0}^{t}ds\, e^{i\Delta(t-s)}\left[W\left(\left\langle x\right\rangle _{s}-\cdot\right)\right]\varphi_{s}\right\Vert _{2}\\
 & \leq & \int_{0}^{t}ds\left\Vert \left[W\left(\left\langle x\right\rangle _{s}-\cdot\right)\right]\left(\varphi_{s}-e^{i\Delta s}\phi^{(0)}\right)\right\Vert _{2}\\
 &  & +\int_{0}^{t}ds\left\Vert \left[W\left(\left\langle x\right\rangle _{s}-\cdot\right)\right]e^{i\Delta s}\phi^{(0)}\right\Vert _{2}\\
 & \leq & \left\Vert W\right\Vert _{\infty}\int_{0}^{t}ds\left\Vert \varphi_{s}-e^{i\Delta s}\phi^{(0)}\right\Vert _{2}+t\left\Vert W\right\Vert _{2}\sup_{s\in[0,t]}\left\Vert e^{i\Delta s}\phi^{(0)}\right\Vert _{\infty}.
\end{eqnarray*}
Inequality (\ref{eq:free estimate}) and Grönwall's Lemma ensure the existence of some $C_{\mathrm{prop}}\in\mathtt{Bounds}$ such that
\begin{equation}
\left\Vert \varphi_{t}-e^{i\Delta t}\phi^{(0)}\right\Vert _{2}\leq C_{\mathrm{prop}}(t)\left|\Lambda\right|^{-1/2}.\label{eq:gronwall}
\end{equation}
Estimates (\ref{eq:free estimate}) and (\ref{eq:gronwall}) prove
our claim.
\item Equation (\ref{eq:free_reference}) and (\ref{eq:initial_reference}), 
together with the estimate (\ref{eq:free estimate}), implies that
\[
\left\Vert \phi_{t}^{(\mathrm{ref})}\right\Vert _{\infty}=\left|\Lambda\right|^{1/2}\left\Vert e^{i\Delta t}\phi^{(0)}\right\Vert _{\infty}\leq C.
\]

\end{enumerate}

Similarly, with (\ref{eq:initial_gas_wf}), one finds that
\[
\left\Vert \nabla\phi_{t}^{(\mathrm{ref})}\right\Vert _{\infty}=\left|\Lambda\right|^{1/2}\left\Vert e^{i\Delta t}\nabla\phi^{(0)}\right\Vert _{\infty}\leq\left|\Lambda\right|^{1/2}\left\Vert \widehat{e^{i\Delta t}\nabla\phi^{(0)}}\right\Vert _{1}=\left|\Lambda\right|^{1/2}\left\Vert \widehat{\nabla\phi^{(0)}(k)}\right\Vert _{1}\leq C\left|\Lambda\right|^{-1/3}.
\]

\end{proof}
\begin{lem}
\label{lem:ref_and_excitation}Let $t\mapsto\epsilon_{t}$ be the
solution to (\ref{eq:macro_gas}) with initial data $\epsilon_{t}|_{t=0}=0$.
There are $C_{\epsilon},\widetilde{C}_{\epsilon}\in\mathtt{Bounds}$
such that, for all times $t\geq0$, the following estimates hold:
\begin{enumerate}[label=(\roman*)]
\item $\left\Vert \epsilon_{t}\right\Vert _{2}\leq C_{\epsilon}(t).$
\item $\left\Vert p_{t}^{(\mathrm{ref})}\epsilon_{t}\right\Vert_{2} \leq\frac{\widetilde{C}_{\epsilon}(t)}{\left|\Lambda\right|^{1/2}}$.
\end{enumerate}
\end{lem}
\begin{proof}
~
\begin{enumerate}[label=(\roman*)]
\item Since the homogeneous part of (\ref{eq:macro_gas}) is self-adjoint
we may infer control over the norm of a solution $t\mapsto\epsilon_{t}$
from the inhomogeneity according to
\begin{eqnarray*}
\left\Vert \epsilon_{t}\right\Vert  & \leq & \int_{0}^{t}\left\Vert W(X_{s}-\cdot)\phi_{s}^{(\mathrm{ref})}\right\Vert _{2}ds\\
 & \leq & t\,\left\Vert W\right\Vert _{2}\sup_{s\in[0,t]}\left\Vert \phi_{s}^{(\mathrm{ref})}\right\Vert _{\infty}\\
 & \leq & t\, C\, C_{\mathrm{ref}}\\
 & =: & C_{\epsilon}(t).
\end{eqnarray*}

\item Using the equations of motion (\ref{eq:free_reference}) and (\ref{eq:macro_gas}),
as well as initial condition (\ref{eq:semi-classical initial values}), a
direct computation yields
\begin{eqnarray*}
\left|\left\langle \frac{\phi_{t}^{(\mathrm{ref})}}{\left|\Lambda\right|^{1/2}},\epsilon_{t}\right\rangle \right| & \leq & \int_{0}^{t}\left|-i\left\langle \frac{\phi_{s}^{(\mathrm{ref})}}{\left|\Lambda\right|^{1/2}},W(X_{s}-\cdot)\epsilon_{s}\right\rangle -i\left\langle \frac{\phi_{s}^{(\mathrm{ref})}}{\left|\Lambda\right|^{1/2}},W(X_{s}-\cdot)\phi_{s}^{(\mathrm{ref})}\right\rangle \right|ds\\
 & \leq & \frac{t}{\left|\Lambda\right|^{1/2}}\sup_{s\in[0,t]}\left\Vert \phi_{s}^{(ref)}\right\Vert _{\infty}\left\Vert W\right\Vert _{2}\sup_{s\in[0,t]}\left\Vert \epsilon_{s}\right\Vert _{2}+\frac{t}{\left|\Lambda\right|^{1/2}}\sup_{s\in[0,t]}\left\Vert \phi_{s}^{(ref)}\right\Vert _{\infty}^{2}\left\Vert W\right\Vert _{1}\\
 & \leq & \frac{t}{\left|\Lambda\right|^{1/2}}C_{\mathrm{ref}}\, C\, C_{\epsilon}(t)+\frac{t}{\left|\Lambda\right|^{1/2}}C_{\mathrm{ref}}^{2}\, C\\
 & =: & \frac{\widetilde{C}_{\epsilon}(t)}{\left|\Lambda\right|^{1/2}},
\end{eqnarray*}
which holds because of (\ref{eq:phi_ref_prop_est}) in \lemref{proagation}.
\end{enumerate}
\end{proof}

\bibliographystyle{alpha}

\vskip.5cm

\noindent \emph{Dirk-Andr\'e Deckert}\\
Department of Mathematics\\
University of California Davis\\
One Shields Avenue, Davis, California 95616, USA\\
\texttt{deckert@math.ucdavis.edu}

\vskip.5cm

\noindent \emph{J\"urg Fr\"ohlich}\\
School of Mathematics\\
The Institute for Advanced Study\\
Princeton, NJ 08540, USA\\
\texttt{juerg@phys.ethz.ch}

\vskip.5cm

\noindent \emph{Peter Pickl}\\
Mathematisches Institut der LMU M\"unchen\\
Theresienstra\ss e 39, 80333 M\"unchen, Germany\\
\texttt{pickl@math.lmu.de}

\vskip.5cm

\noindent \emph{Alessandro Pizzo}\\
Department of Mathematics\\
University of California Davis\\
One Shields Avenue, Davis, California 95616, USA\\
\texttt{pizzo@math.ucdavis.edu}

\end{document}